\documentclass[12pt]{article}
\usepackage{amsmath,amssymb,amsfonts,amsthm,url,xspace,graphicx,enumerate}
\newtheorem{theorem}{Theorem}
\numberwithin{theorem}{section}

\newtheorem{lemma}[theorem]{Lemma}
\newtheorem{proposition}[theorem]{Proposition}
\newtheorem{corollary}[theorem]{Corollary}
\newtheorem{definition}[theorem]{Definition}
\newtheorem*{thm}{Theorem}
\theoremstyle{remark}
\newtheorem*{unremark}{Remark}

\newtheorem{example}[theorem]{Example}

\newcommand{\Z}{\mathbb{Z}}
\newcommand{\R}{\mathbb{R}}
\newcommand{\RP}{\mathbb{RP}}
\newcommand{\CP}{\mathbb{CP}}
\newcommand{\C}{\mathbb{C}}
\newcommand{\G}{\mathcal{G}}

\newcommand{\A}{{\cal A}}

\newcommand{\be}{\begin{equation}}
\newcommand{\ee}{\end{equation}}
\newcommand{\old}[1]{}

\newcommand{\Em}[1]{\textbf{#1}}

\usepackage[usenames,dvipsnames]{color}

\def\noproof{\hfill \Box}
\def\init{{\cal I}}
\def\lat{{\cal L}}
\def\flabel{{\Z^3_{1/2}}}
\def\ising{{\Upsilon}}

\begin{document}
\title{Double-dimers, the Ising model and the hexahedron recurrence}
\author{R. Kenyon and R. Pemantle}
\date{}
\maketitle

\begin{abstract}
We define and study a recurrence relation in $\Z^3$, called 
the hexahedron recurrence, which is similar to the octahedron 
recurrence (Hirota bilinear difference equation) and cube recurrence 
(Miwa equation).  Like these examples, solutions to the 
hexahedron recurrence are partition sums for 
edge configurations on a certain graph, and have a natural 
interpretation in terms of cluster algebras.  We give an 
explicit correspondence between monomials in the Laurent 
expansions arising in the recurrence with certain double-dimer 
configurations of a graph.  We compute limit shapes for the 
corresponding double-dimer configurations. 

The Kashaev difference equation arising in the Ising model 
star-triangle relation is a special case of the hexahedron 
recurrence.  In particular this reveals the cluster nature 
underlying the Ising model. The above relation allows us to 
prove a Laurent phenomenon for the Kashaev difference equation. 
\end{abstract}

\section{Introduction}

\subsection{Overview}

A function $a:\Z^3\to\C$ is said to satisfy
the \emph{octahedron recurrence} or \emph{Hirota bilinear difference equation}  if at all points $v\in\Z^3$
\be\label{eq:oct}
a_{(1)}a_{(23)}=a_{(2)}a_{(13)}+a_{(3)}a_{(12)}\ee
(here $a_{(S)}\equiv a_{v+(S)}$ represents $a$ evaluated at the translate of $v$ by the basis vectors in $S$, e.g. $a_{(12)}$ represents $a_{v+e_1+e_2}$). 
The octahedron recurrence was coined by Propp (see \cite{Speyer}) but first appeared in 
Dodgson~\cite{Dodg1866} as a means of recursively computing 
determinants: up to an affine change of indices (and some sign changes) it is the recurrence
satisfied by determinants of contiguous submatrices. In this setting it is known
as \emph{Dodgson condensation}.

The octahedron recurrence is fundamental in combinatorics, statistical mechanics,
cluster algebras, and integrable systems, see e.g. \cite{KNS}.

There are several similar recurrences. 
The most well-known is the \emph{cube recurrence} or \emph{Miwa equation}.
A function $g:\Z^3\to\C$ is said to satisfy the Miwa equation or cube recurrence if
\be\label{eq:cube}
g_{(123)}g=g_{(1)}g_{(23)}+g_{(2)}g_{(13)}+g_{(3)}g_{(12)}.\ee
This recurrence also has its roots in the $19^{th}$ century: Kennelly~\cite{Kenn1899}
discovered the so-called star-triangle identity (wye-delta transformation) for resistor networks.
Under a certain change of variables (see \cite{FZ, GK}) this transformation can be written as a cube recurrence.

In \cite{FZ}, see also \cite{GK}, it was noticed that the cube recurrence is a specialization 
of a more fundamental recurrence, which we call the \emph{cuboctahedron recurrence}.
This is a recurrence on a function on the edges of the cubes of the $\Z^3$ lattice. If certain monomial equations are
satisfied then the cuboctahedron recurrence reduces to the cube recurrence. Like the octahedron recurrence,
the cuboctahedron recurrence arises from cluster algebras, and in fact is a composition of cluster algebra mutations
on a certain planar graph.
This leads to a cluster algebra interpretation of the cube recurrence, and in particular allows one to prove a Laurent phenomenon
as shown in \cite{FZ}.

A less-well known recurrence is the \emph{Kashaev recurrence} \cite{Kash}. 
A function $f:\Z^3\to\C$ is said to satisfy the Kashaev recurrence if
\begin{eqnarray}\label{KashaevEQ}
f^2f_{(123)}^2+f_{(1)}^2f_{(23)}^2+f_{(2)}^2f_{(13)}^2+f_{(3)}^2f_{(12)}^2
   -2f_{(1)}f_{(2)}f_{(23)}f_{(13)} - 2f_{(1)}f_{(3)}f_{(23)}f_{(12)}-2f_{(3)}f_{(2)}f_{(12)}f_{(13)} \\\nonumber
   - 2ff_{(123)}(f_{(1)}f_{(23)} +f_{(2)}f_{(13)}+f_{(3)}f_{(12)}) 
   - 4ff_{(23)}f_{(13)}f_{(12)}-4f_{(123)}f_{(1)}f_{(2)}f_{(3)} 
&=&0. \nonumber
\end{eqnarray}
This recurrence arises in the star-triangle move (Yang-Baxter equation) for the Ising model.

Our main goal in this paper is to define another recurrence, generalizing the Kashaev recurrence, called
the \emph{hexahedron recurrence}\footnote{The hexahedron is another name for the cube,
but emphasizing the fact that it has $6$ faces. We chose this nomenclature since our variables sit on both the vertices
and faces of a cube.}.
This is a relation for a function defined on the faces and vertices of the $\Z^3$ cubic tiling.
Given four functions 
$$h,h^{(x)},h^{(y)},h^{(z)}\colon\Z^3\to\C$$
(where we think of $h$ as being the function on the vertices of the cubes, 
$h^{(x)}_v$ as being the value of the function on the ``yz"-face with vertices $\{v,v+e_2,v+e_2+e_3,v+e_3\}$, and similarly
for $h_v^{(y)}$ and $h_v^{(z)}$ on the ``xz" and ``xy" faces respectively)
we say they satisfy the 
hexahedron recurrence if the following equations are satisfied for all $v\in\Z^3$.

{\small \begin{eqnarray}\label{hh1}
h^{(x)}_{(1)}h^{(x)}h&=&h^{(x)}h^{(y)}h^{(z)}+h_{(1)}h_{(2)}h_{(3)}+hh_{(1)}h_{(23)}\\
\label{hh2}
h^{(y)}_{(2)}h^{(y)}h&=&h^{(x)}h^{(y)}h^{(z)}+h_{(1)}h_{(2)}h_{(3)}+hh_{(2)}h_{(13)}\\
\label{hh3}
h^{(z)}_{(3)}h^{(z)}h&=&h^{(x)}h^{(y)}h^{(z)}+h_{(1)}h_{(2)}h_{(3)}+hh_{(3)}h_{(12)}\\
h_{(123)}h^2h^{(x)}h^{(y)}h^{(z)}&=&
(h^{(x)}h^{(y)}h^{(z)})^2+h^{(x)}h^{(y)}h^{(z)}(2h_{(1)}h_{(2)}h_{(3)}+\nonumber
hh_{(1)}h_{(23)}+hh_{(2)}h_{(13)}+hh_{(3)}h_{(12)})+\\&&\label{hh4}
+(h_{(1)}h_{(2)}+hh_{(12)})(
h_{(1)}h_{(3)}+hh_{(13)})(h_{(2)}h_{(3)}+hh_{(23)}).
\end{eqnarray}}

Here again $h_{(1)}= h_{v+e_1}$ and so on.
In Section~\ref{Ising} below we will show that the Kashaev recurrence 
is a special case of the 
hexahedron recurrence. 
Understanding the cluster algebra structure of the 
Y-Delta transformation for the Ising model was our initial motivation for defining the
hexahedron recurrence.

The octahedron, cuboctahedron and hexahedron recurrences have an underlying cluster algebra structure, based on the local transformation (mutation)
called \emph{urban renewal}, see Figure \ref{Aurban}. For example the octahedron recurrence can be thought of
as a ``periodic" urban renewal step on the square grid graph; the cuboctahedron recurrence can be decomposed into
a product of $4$ periodic urban renewal steps, see \cite{FZ}.
The hexahedron recurrence can be decomposed into a product of $6$ periodic urban renewals. 
In particular this allows us to write $h_{i,j,k},h^{(x)}_{i,j,k},h^{(y)}_{i,j,k},h^{(z)}_{i,j,k}$ for $i+j+k$ large as Laurent polynomials
in the initial values $\{h_{i,j,k}\}_{0\le i+j+k\le 2}$ and $\{h^{(x)}_{i,j,k},h^{(y)}_{i,j,k},h^{(z)}_{i,j,k}\}_{i+j+k=0}$.

Speyer \cite{Speyer} and Carroll and Speyer \cite{CS} respectively gave combinatorial
interpretations of the terms in the Laurent expansions arising in the octahedron and cube recurrences. 
One of our main results is an analogous result for the hexahedron recurrence: we give an 
explicit bijection between the monomials of the Laurent polynomials which arise and certain
``taut" double-dimer covers of a sequence of planar graphs $\Gamma_n$. 

Following the method of Peterson and Speyer \cite{PS} we also prove a limit shape theorem for random taut 
double-dimer covers of $\Gamma_n$
with an arctic boundary which is an algebraic curve defined from the characteristic polynomial of the
recurrence relation. 

The integrable nature of these systems will not be discussed here (although it is easy to show that they
satisfy a standard multidimensional consistency). The hexahedron recurrence is 
a special case of a dimer integrable system, and integrable properties of such systems are discussed in \cite{GK}. 

\subsection{Results}
We may picture
the values of $h^{(x)}$, $h^{(y)}$ and $h^{(z)}$ as each sitting 
in the middle of a square face of the integer lattice.  Letting 
$$\flabel := \{ (x,y,z) \in (1/2) \Z^3 : x + y + z \in \Z \}$$
we then interpret $h^{(x)} , h^{(y)}, h^{(z)}$ as extending
$h$ to $\flabel$ via $h^{(x)} (i,j,k) = h(i,j+1/2,k+1/2)$,
$h^{(y)} (i,j,k) = h(i+1/2,j,k+1/2)$, and $h^{(z)} (i,j,k) 
= h(i+1/2,j+1/2,k)$.  

Note that the 
equations (\ref{hh1}-\ref{hh4}) are not only homogeneous, but are 1-homogeneous:
the sum of all indices of each monomial is $3$, e.g., the
first monomial is the product of three variables with 
respective indices $(1,1/2,1/2)$, $(0,1/2,1/2)$ and $(0,0,0)$.
Also, given the values of $h$ on seven of the corners and their 
three included faces of a cube, the values on the eighth corner 
and the three remaining faces are determined as rational functions 
of these; the locations of the new values are precisely those
obtained if one increases a pile by a single cube.

The cluster nature of the hexahedron recurrence 
immediately implies:
\begin{thm}
The values $\{ h(v) : v \in \flabel, v_1 + v_2 + v_3 \geq 0 \}$
are Laurent polynomials in the values of $h(v)$ such that 
$v$ is an integer vector and $0 \leq v_1 + v_2 + v_3 \leq 2$ or
$v$ is a half-integer vector and $v_1 + v_2 + v_3 = 1$.
\end{thm}
Our main result, Theorem~\ref{maincomb} below, gives a combinatorial
interpretation of the terms of these Laurent polynomials.  

\subsection{Statistical mechanical interpretations}

The octahedron recurrence may be used to express $a_v$ as a Laurent
polynomial in the values $a_w$ as $w$ ranges over variables in an
initial graph.  When the initial graph is the grid $\Z^2$ this Laurent polynomial is a 
generating function for a certain statistical mechanical ensemble: 
its monomials are in bijection with perfect matchings of the 
\Em{Aztec diamond} graph, associated with the region in the initial  
graph lying in the shadow of $v$; see, e.g.,~\cite{Speyer}.
Setting the initial indeterminates
all equal to one allows us to count perfect matchings; in general
the indeterminates represent multiplicative weights, which we may 
change in certain natural ways to study further properties of the
ensemble of perfect matchings.  These ensembles, well studied on
the square lattice, exist on many other periodic bipartite planar graphs.

The cube recurrence~\eqref{eq:cube} also has a combinatorial 
interpretation.  Its monomials are in bijection with \Em{cube groves}.  
These were first defined and studied in~\cite{CS,PS} where they were called simply ``groves".
In a cube grove, each edge of a large triangular region in the 
planar triangular lattice is either present or absent.  The
allowed configurations are those in which the edge subsets contain no cycles
and no islands (thus they are \Em{essential spanning forests}), 
and the connectivity of boundary points has a prescribed form.  

Both
Aztec diamond matchings and cube groves have limiting
shapes.  Specifically, as the size of the box goes to infinity,
there is a boundary, which is an algebraic curve, outside
of which there is no entropy (the system is periodic almost surely)
and inside of which there is
positive entropy per site (each type of local configuration occurs with positive probability).  
In the case of the Aztec diamond
and cube grove, the algebraic curve is the inscribed circle,
but for related ensembles, much more general algebraic curves are
obtained; for example in the \Em{fortress} tiling model
shown on the right of Figure~\ref{fig:fortress}, the bounding
curve is a degree-8 algebraic curve\footnote{Dubbed the ``octic circle''
by the fun-loving pioneers of this subject.}, see \cite{KO}.
\begin{figure}[htbp]
\centering
\includegraphics[scale=0.3, trim = 1cm 6cm 1cm 2cm, clip=true]{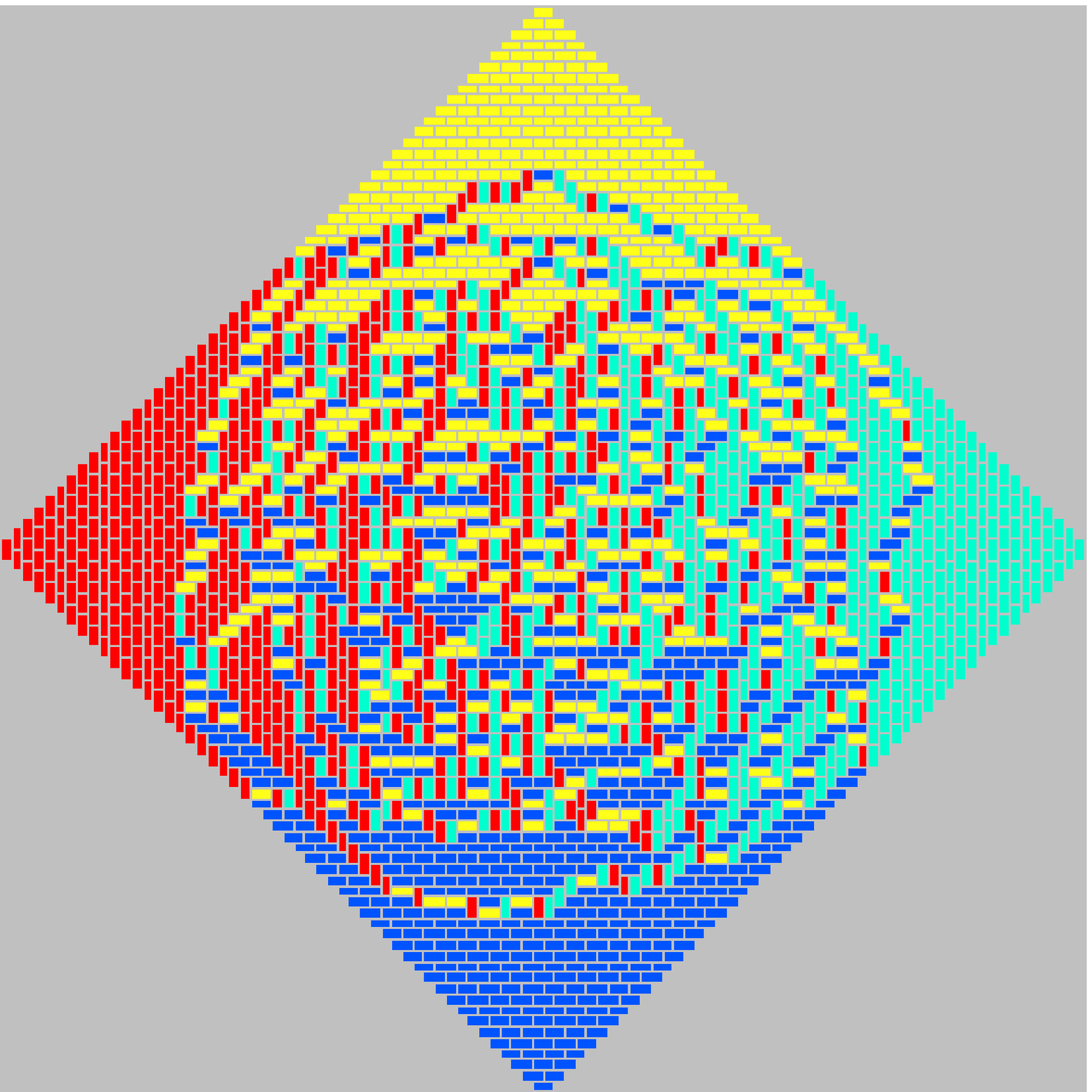}
\includegraphics[scale=0.38]{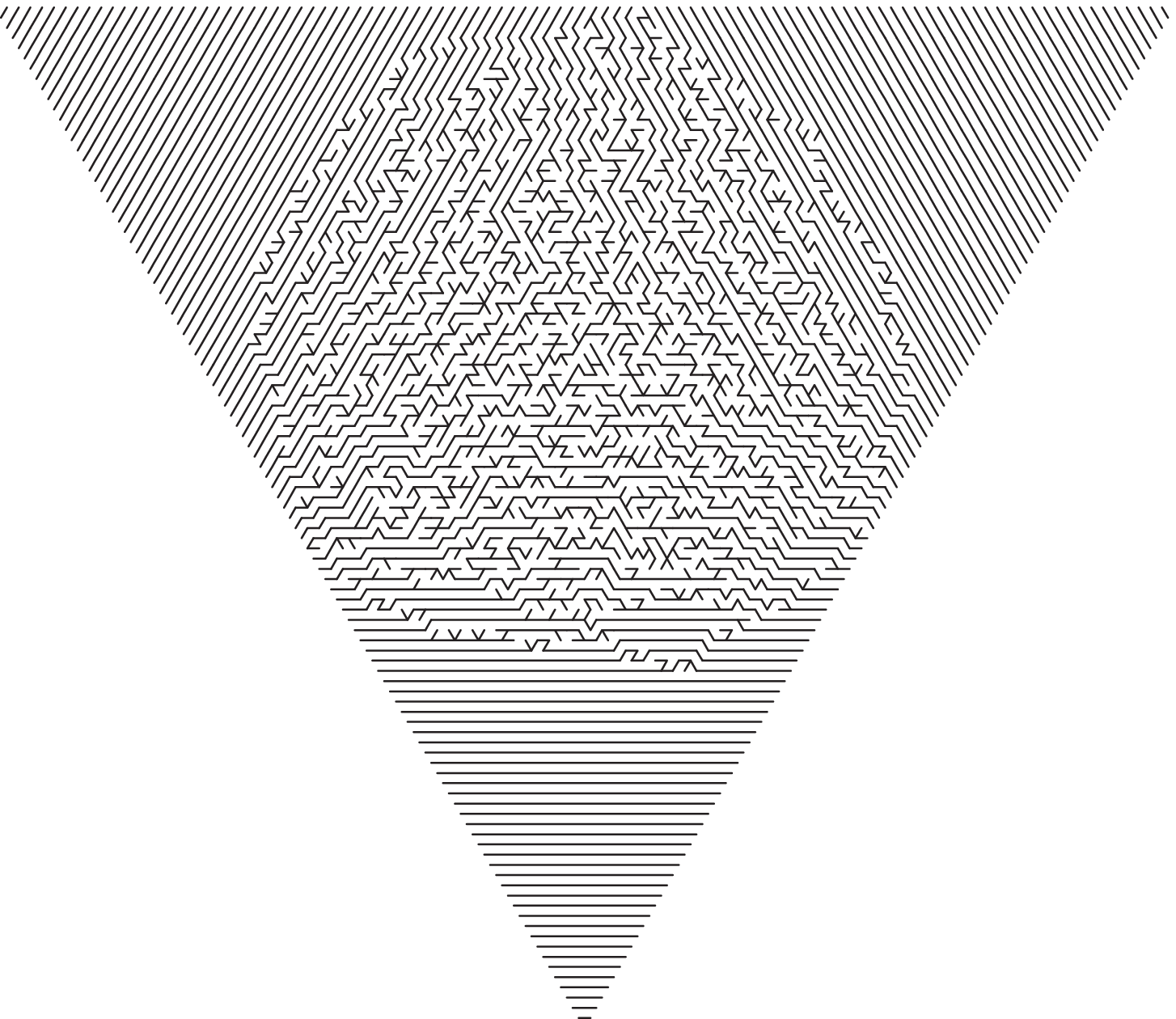}
\includegraphics[scale=0.25]{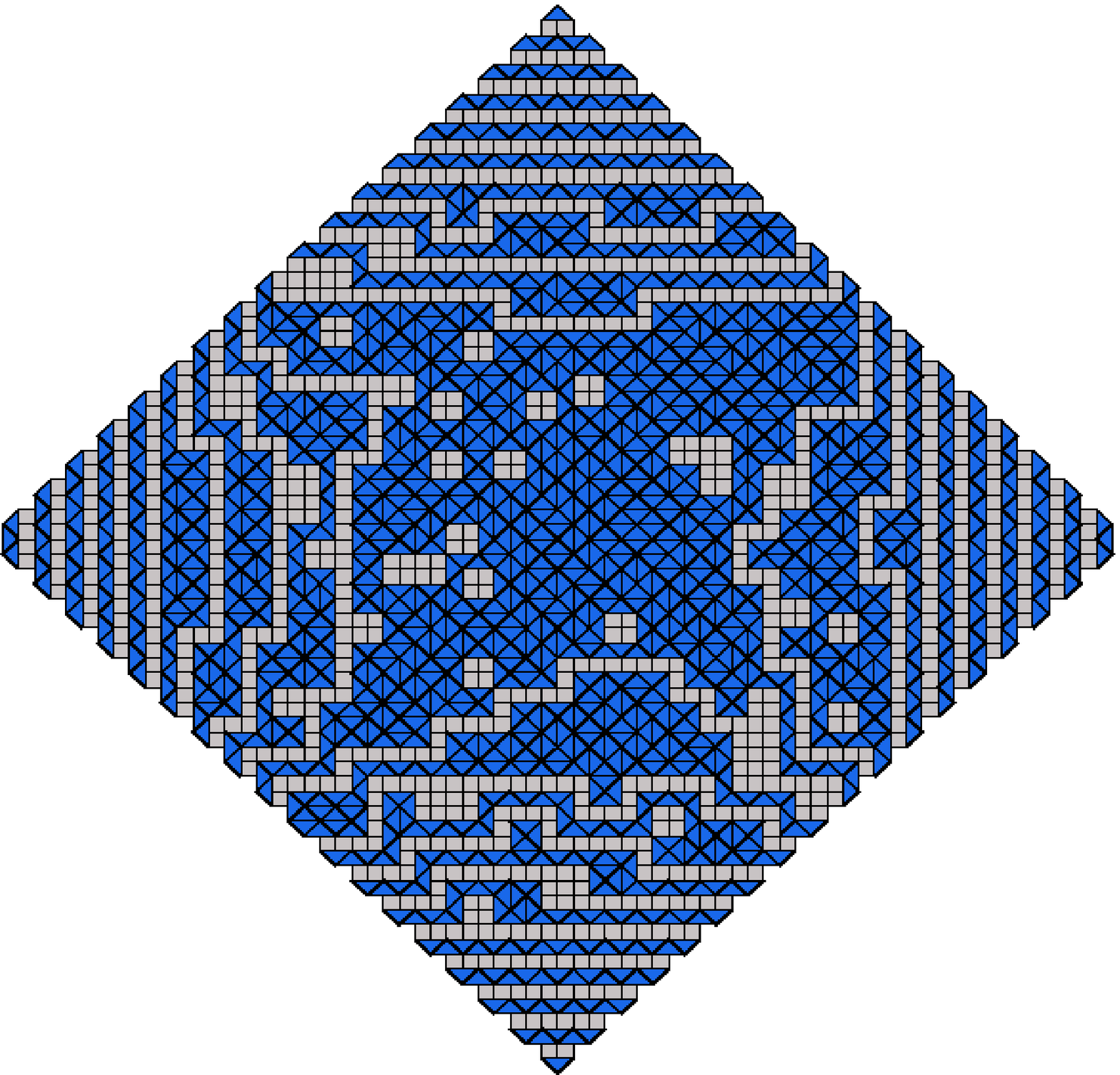}
\caption{An Aztec tiling, a cube grove and a fortress tiling}
\label{fig:fortress}
\end{figure}

The hexahedron recurrence also has a statistical mechanical
interpretation.  In Section~\ref{sec:dimers} we define the 
\Em{double-dimer} model on a finite bipartite graph.  Terms in the hexahedron recurrence
count certain types of double-dimer configurations called \emph{taut} configurations. A randomly 
generated taut double-dimer configuration is shown in Figure~\ref{fig:dd}.

\begin{figure}[htbp]
\centering
\includegraphics[scale=0.6]{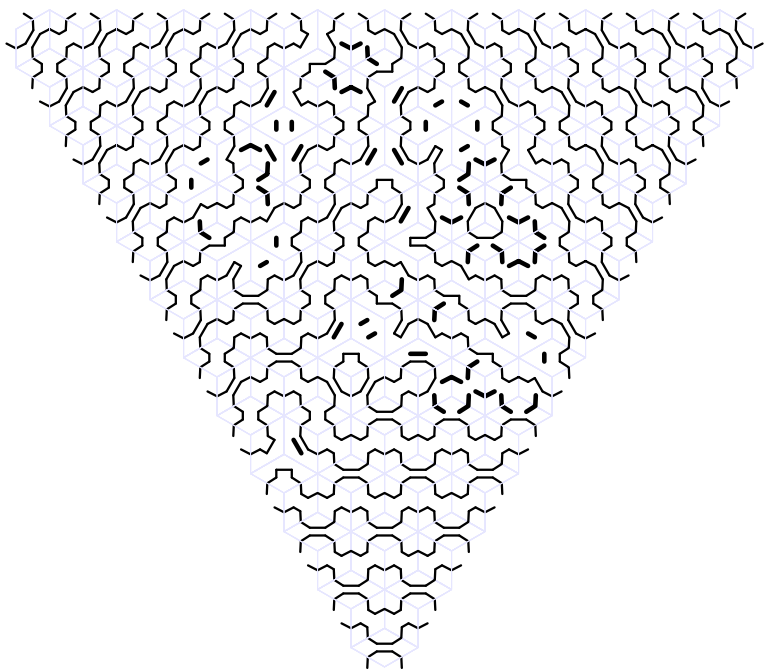}
\caption{A random taut double-dimer configuration}
\label{fig:dd}
\end{figure}

In Section~\ref{sec:main} we prove the following theorem.
\begin{thm}
The monomials of the Laurent polynomial $h_{iii}$ are in bijection
with taut double-dimer configurations of the graph $\Gamma_{3i}$. 
\end{thm}
The remainder of the paper is spent investigating the properties
of the hexahedron double-dimer ensemble.  In Section~\ref{ddlimitshape} we analyze 
the limiting shape under several natural, periodic specifications
of the initial variables.  In Section~\ref{Ising} we find a
specialization of the initial variables under which the hexahedron recurrence
becomes the Ising $Y$-$\Delta$ transformation.
\old{, which
is a transformation of the Ising model changing the interaction
strengths in an algebraically different way from how they change under
the resistor-network $Y$-$\Delta$ transformation, but changing the
graph in exactly the same way.  }

\bigskip
\noindent{\bf Acknowledgements.} We would like to thank Philippe Di Francesco, Sergei Fomin,
David Speyer and Dylan Thurston for helpful discussions.

\section{Dimer model} \label{sec:dimers}

\subsection{Definitions}

We will use the following graph theoretic terminology.  A
\Em{bipartite} graph is a graph $\Gamma = (E,V)$ together
with a fixed coloring of the vertices into two colors (black and white)
such that edges connect only vertices of different colors.
By \Em{planar} graph we mean one with a distinguished
planar embedding. The \emph{dual graph} $\Gamma_*$ 
to a planar bipartite graph $\Gamma$ has a vertex $f_*$ for
each face $f$ of $\Gamma$ and edge $e_*$ connecting $f_*$ 
and $g_*$ where $f$ and $g$ are the two faces containing $e$.

Let $\Gamma$ be a finite bipartite graph with positive edge weights 
$\nu : E \to \R_+$.  A \Em{dimer cover} or \Em{perfect matching} 
is a collection of edges with the property that every vertex 
is an endpoint of exactly one edge.  The ``dimers'' of a dimer
cover are the chosen edges (terminology suggesting a collection
of bi-atomic molecules packed into the graph).  We let 
$\Omega_d (\Gamma)$ be the set of dimer covers and we define 
the probability measure $\mu_d$ on $\Omega_{d}$ giving a dimer cover 
$m \in \Omega_d$ a probability proportional to the product of its 
edge weights.

A \Em{double-dimer configuration} is a union of two dimer covers: 
it is a covering of the graph with loops and doubled edges.
The \Em{double dimer measure} $\mu_{dd}$ is the probability measure 
defined by taking the union of two $\mu_d$-independent dimer covers.

\subsection{Edge variables, face variables, and cluster variables}

\subsubsection{Gauge transformations}

A gauge transformation consists in multiplying the edge weights 
of edges incident to a vertex $v$ by a constant.  This leaves 
$\mu_d$ unchanged since every dimer cover using exactly one edge 
incident to $v$.

The gauge transformations form a group isomorphic to $\R_+^{|V|-1}$ 
(multiplying edges at all white vertices by $\lambda$ and edges at 
all black vertices by $\lambda^{-1}$ acts trivially) and the quotient 
space of edge weights modulo gauge transformations is isomorphic to 
$\R_+^c$, $c$ being the dimension of the cycle space of $\Gamma$.  
For a planar graph (meaning a specific planar embedding has been chosen),
whose set of faces is denoted $F$, it is natural to coordinatize 
the space of edge weights modulo gauge with variables
$\{X_f\}_{f\in F}$; here for a face $f$, $X_f$ is the alternating 
product of the edge weights around that face: orient all edges 
from white to black, then take the product over edges $e$ in the 
face of $\nu (e)^{\delta (e)}$ where $\delta = 1$ if the edge $e$ 
points counterclockwise with respect to the planar embedding and 
$\delta (e) = -1$ if the edge $e$ points clockwise. These $X_f$-variables are 
called \emph{cross-ratio variables}\footnote{Historically we use notation 
$X$ for these variables although from the point of view of cluster algebras
these are ``coefficient variables" and are represented with $Y$s.}. They have no relations (we do not assign
a variable to the outer face).

We also will need a third set of variables $\{A_f\}_{f\in F}$ called \emph{cluster variables}, 
also living on faces of $\Gamma$. The relationship with the edge variables is as follows.
Given $A : F \to \C$, define $\nu_A : E \to \C$ by
\begin{equation} \label{eq:A-e}
\nu(e) = \frac{1}{A_f A_g} \, ,
\end{equation}
where $f$ and $g$ are the two faces containing $e$.  It may 
be checked that for any face $f \in F$, 
$$X_f = \prod_{e \in F} A_f^{\delta(e,f)},$$ where $\delta(e,f)=\pm1$ according to whether the dual edge $e^*f^*$
sees a white vertex on the right or left. 
We summarize this in a commuting diagram.  

\setlength{\unitlength}{1.5pt}
\begin{picture}(220,120)(0,-20)
\put(10,60){$\{ A \}$}
\put(90,60){$\{ e \}$}
\put(50,00){$\{ X \}$}
\put(40,60){\vector(1,0){40}}
\put(30,50){\vector(1,-2){18}}
\put(90,50){\vector(-1,-2){18}}
\put(40,75){$\displaystyle{\frac{1}{A(f) A(g)}}$}
\put(20,34){alt.} 
\put(20,26){prod.} 
\put(90,34){alt.} 
\put(90,26){prod.} 
\put(70,30){$\pi$}
\put(45,31){$\rho$}
\put(150,45){$\pi$ is surjective,} 
\put(150,35){kernel of $\pi$ is the gauge functions}
\put(150,15){$\rho$ is not surjective}
\end{picture}

Not all configurations of $X$ variables are in the range of $\rho$.
The quotient is multiplicatively generated by functions 
which are $1$ on one side of a zigzag path and $\lambda$ on
the other side.  For the purposes of this paper we will not consider more general $X$ variables other
than those in the range of $\rho$.

\subsection{Urban renewal}  \label{ss:renewal}

Certain local rearrangements of a bipartite graph $\Gamma$ preserve the dimer measure 
$\mu_d$, see \cite{KPW}.  These are: parallel edge reduction, vertex contraction/splitting,
and urban renewal.

The simplest is \emph{parallel edge reduction}. If there are two parallel edges of $\Gamma$ (edges with the same endpoints)
with weights $a$ and $b$ then we can replace these with a single edge of weight $a+b$. Clearly there is a coupling
of the dimer measure of the ``before" graph and ``after" graph.

\emph{Vertex contraction} is described as follows. Given 
a vertex $v$ of degree $2$, with equal weights on its two edges, one can contract 
its two edges, erasing $v$ and merging its two neighbors into 
one vertex.  The faces of the new graph are in bijection with
the faces of the old graph, the only difference being that
two faces have each lost two consecutive equally weighted edges.
The
dimer measure $\mu$ on the original weighted graph $\Gamma$ 
may be coupled to the new dimer measure $\mu'$ on the new 
contracted graph $\Gamma'$ as follows.  To sample from $\mu'$, sample 
a matching $m$ from $\mu$ and then delete whichever edge 
of $m$ contains $v$; the new set of dimers, $m'$ will be
a perfect matching on $\Gamma'$ and it is obvious that
the law of $m'$ is $\mu'$.  The inverse of this contraction
operation, splitting a vertex in two and adding a vertex of 
degree $2$ between them (with equal edge weights on the two 
edges) is called \emph{vertex splitting}: a vertex $w$ is split into
two vertices $w_1$ and $w_2$, each incident to a proper
subset of the vertices formerly incident to $w$ (this subset
being an interval in the cyclic order induced by the 
planar embedding); a new vertex $v$ is introduced whose
neighbors are $w_1$ and $w_2$.  The new planar embedding
is obvious. To sample from the new measure $\mu'$, sample
from $\mu$ and then add either $v w_1$
or $v w_2$ depending on which vertex $w_1,w_2$ is not matched.

The more interesting local rearrangement is called \Em{urban renewal}. 
It involves taking a quadrilateral face, call it~0, and adding 
``legs".  This is shown in Figure \ref{Aurban}, ignoring for the
moment the specific values $a_0, \ldots , a_4$ shown for the 
pre-weights $A(0) , \ldots , A(4)$.  
\begin{figure}[htbp]
\center{\includegraphics[width=3in]{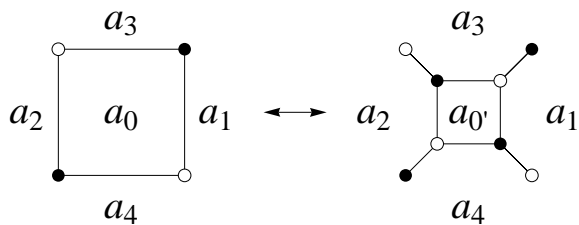}}
\caption{\label{Aurban}Under urban renewal the central variable $a_0$ 
changes from $a_0$ to $a_{0'}=\frac{a_1a_2+a_3a_4}{a_0}$.}
\end{figure}
Let us designate the faces around 
face~0 by the numbers 1, 3, 2 and 4.  Each of these faces gains two 
new edges.  In the new graph $\Gamma'$, there are faces $1',
2', 3'$ and $4'$ each with two more edges than the corresponding
face $1, 2, 3, 4$.  There is a face $0'$ which is also square.
Each other face $f$ of $\Gamma$ corresponds to a face $f'$
of $\Gamma'$ with the same number of edges as $f$.  There are
four new neighboring relations among faces: $1', 2', 3'$ and $4'$
are neighbors in cyclic order, in addition to any neighboring
relations that may have held before.  The point of urban renewal
is to give a corresponding adjustment of the weights that
preserves $\mu_d$.  This is most easily done in terms of
the $A$ variables.  The following proposition was proved
in~\cite{Speyer}.
\begin{proposition}[urban renewal]
Suppose~0 is a quadrilateral face of $\Gamma$.  Let $\Gamma'$
be constructed from $\Gamma$ as above.  Define the new pre-weight
function $A: F' \to \C$ by $A(f') = A(f)$ if $f \neq 0$ and 
$$A(0') := \frac{ A(1) A(2) + A(3) A(4)}{A(0)} \, .$$
Let $\mu'$ denote the dimer measure on $\Gamma'$ with edge weights 
$\nu_{A'}$ and $\mu$ the dimer measure on $\Gamma$ with face weights 
$\nu_{A}$.  Then $\mu$ and $\mu'$ may be coupled so that the
sample pair $(m , m')$ agrees on every edge other than the
four edges bounding face~0 in $\Gamma$ and the eight edges
touching face~$0'$ in $\Gamma'$.  
$\noproof$
\end{proposition}

There are several equivalent versions of urban renewal, all of
which are related to each other by vertex contraction and splitting.
Two of these are depicted in Figure \ref{Aurban2}.
\begin{figure}[htbp]
\center{\includegraphics[width=2in]{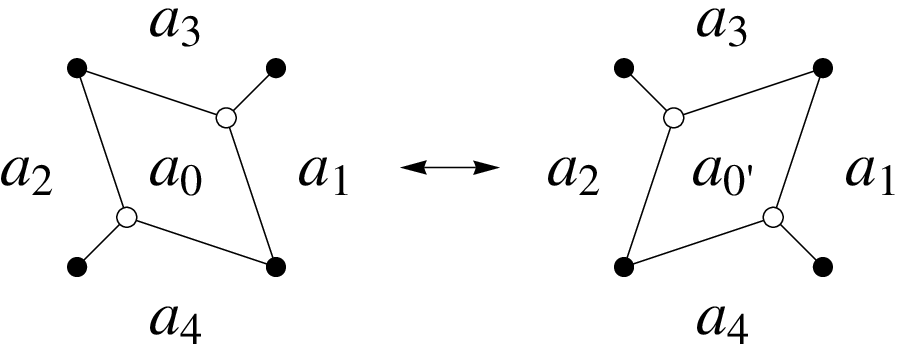}\hskip1in\includegraphics[width=2in]{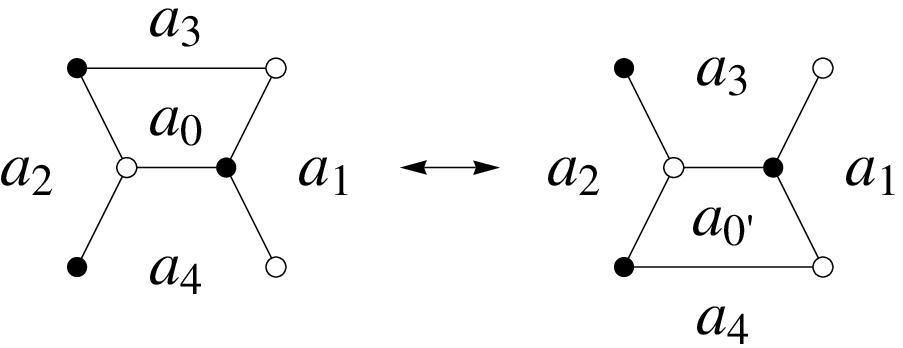}}
\caption{\label{Aurban2}Other variants of urban renewal.}
\end{figure}

\begin{proposition}[\cite{GK}] \label{pr:urban}
The transformation of $(\Gamma , A)$ to $(\Gamma', A')$ under
urban renewal is a mutation operation.  Consequently, the final variables after any number of
urban renewals are Laurent polynomials in the original variables
$\{ A_f : f \in F \}$.
\end{proposition}

For readers familiar with cluster algebras,
the underlying quiver is the dual graph, with edges directed so that white vertices are on the left.

\subsection{Superurban renewal transformation for the dimer model}
A more complicated local transformation is the \emph{superurban renewal} 
shown in Figure \ref{superurban}.  Figure \ref{decompose} shows how this 
can be decomposed into a sequence of six urban renewals.
\begin{figure}[htbp]
\center{\includegraphics[width=3in]{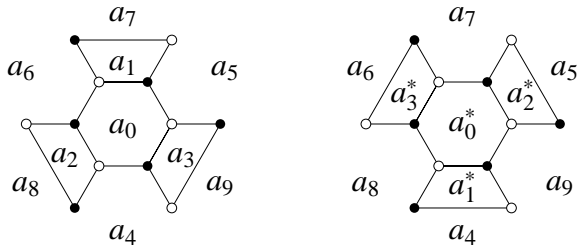}}
\caption{\label{superurban}Under superurban renewal the variables $a_0,a_1,a_2,a_3$ change 
as indicated in (\ref{surec1})-(\ref{superurbanrec}).}
\end{figure}
\begin{figure}[htbp]
\center{\includegraphics[width=6in]{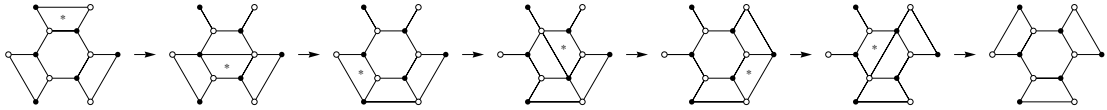}}
\caption{\label{decompose}Writing a superurban renewal as a composition of urban renewals. The stars
indicate which face undergoes urban renewal}
\end{figure}
Just as urban renewal is the basis for the octahedron recurrence,
we will see that superurban renewal is the basis for the hexahedron 
recurrence (see Section~\ref{sec:4-6-12}).  In section \ref{Ising} 
we show how superurban renewal specializes to the Y-Delta 
transformation for the Ising model.

\begin{lemma}\label{sur} Under a superurban renewal the $A$ variables 
transform as in Figure \ref{superurban}, with 
\begin{eqnarray}a_1^* & = & \frac{a_1a_2a_3+a_4a_5a_6+a_0a_4a_7}{a_0a_1}
\   \label{surec1}\\[1ex]
a_2^* & = & \frac{a_1a_2a_3+a_4a_5a_6+a_0a_5a_8}{a_0a_2}
   \label{surec2}\\[1ex]
a_3^* & = & \frac{a_1a_2a_3+a_4a_5a_6+a_0a_6a_9}{a_0a_3}
   \label{surec3}\\[1ex]
a_0^* & = & \frac{a_1^2a_2^2a_3^2 + a_1a_2a_3(2a_4a_5a_6 + a_0a_4a_7 
   + a_0a_5a_8+a_0a_6a_9) + (a_5a_6 + a_0a_7)(a_4 a_5 + a_0 a_9)
   (a_4 a_6 + a_0a_8)}{a_0^2a_1a_2a_3} \, .
   \label{superurbanrec}
\end{eqnarray}
\end{lemma}

\begin{proof}
A computation using the composition of Figure \ref{decompose} and
the urban renewal formula from Figure \ref{Aurban}. 
\end{proof}
Iterating Proposition~\ref{pr:urban} proves the following result.
\begin{corollary}[Laurent property for superurban renewal]
Under iterated superurban renewal, all new variables are
Laurent polynomials in the original variables.
$\noproof$
\end{corollary}

\section{Picturing superurban renewal via stepped surfaces} 
\label{sec:4-6-12}

\subsection{Graphs associated with stepped surfaces}

A \Em{stepped solid} in $\R^3$ is a union $U$ of lattice cubes 
$[i,i+1] \times [j,j+1] \times [k,k+1]$ which is downwardly closed,
meaning that if a cube $B$ is in $U$ then so is any translation
of $B$ by negative values in any coordinate.  In particular,
all our stepped solids are infinite.
A \Em{stepped surface} is the topological boundary of a stepped solid.  
Every stepped surface is the union of lattice squares and every
lattice square has vertex set of the form $\{ v , v + e_i ,
v + e_j , v + e_i + e_j \}$ for some $v \in \Z^3$ and
some integers $1 \leq i < j \leq 3$.  For each stepped surface 
$\partial U$ its $1$-skeleton $\partial U_1$ is a planar graph. We associate to $U$ another graph, the \Em{associated graph} 
$\Gamma (U)$, obtained
by starting with the dual graph of $\partial U_1$ and replacing each
vertex by a small quadrilateral, as illustrated in Figure \ref{fig:4-6-12}.  

Figure~\ref{fig:4-6-12} shows the so-called \Em{4-6-12} graph,
which is the graph $\Gamma (U)$ when $U$ is the union of 
all cubes lying entirely within the region $\{ (x,y,z) : 
x + y + z \leq 2 \}$. 
For general stepped surfaces, the faces of $\Gamma (U)$ will be 
of two types: there is a quadrilateral face centered at the center 
of each face of $\partial U_1$ and there is a $(2k)$-gonal face 
centered at each vertex of $\partial U_1$, where $k$ is the 
number of faces of the surface $\partial U_1$ coming 
together at the vertex, each contributing one edge and two
half-edges (``legs'').  The \Em{label} of the face $f$ is 
the coordinates of the face center or vertex at which $f$
is centered.  All face labels are elements of the set 
$$\flabel := \{ (x,y,z) \in (1/2) \Z^3 : x + y + z \in \Z \} \, .$$
The \Em{canonical variables} are the labels of the 4-6-12 graph,
namely all elements of $\Z^3$ at levels~0, 1 or~2 and all elements
of $(1/2) \Z^3$ at level~1.
\begin{figure}[htbp]
\centering
\includegraphics[scale=0.6]{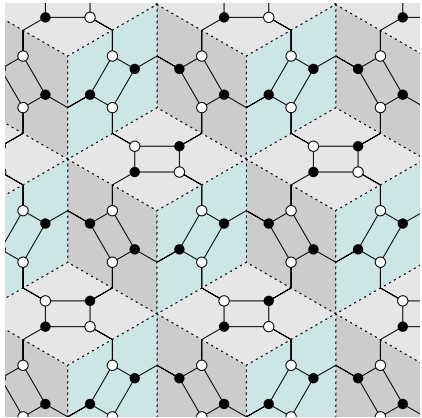}
\caption{The 4-6-12 graph is drawn on the stepped surface 
   $U$ bounding the union of cubes up to level~2}
\label{fig:4-6-12}
\end{figure}

\subsection{Superurban renewal for associated graphs}

Let $U$ be a stepped solid with stepped surface
$\partial U$ and associated graph $\Gamma (U)$.  Suppose that 
$(i,j,k)$ is a point of $\partial U$ which is a local minimum 
with respect to the height function $i+j+k$.  In other words, 
$(i-1,j,k)$, $(i,j-1,k)$ and $(i,j,k-1)$ are all in the interior 
of $U$.  Let $U^{+ijk}$ be the union of $U$ with the cube $[i,i+1]
\times [j,j+1] \times [k,k+1]$.  The following facts are easily 
verified by inspection.

\begin{proposition}[superurban renewal is adding a cube]
\label{pr:super}

Suppose $U$ and $(i,j,k)$ are as above.  
\begin{enumerate}[(i)]
\item The face in $\Gamma (U)$ corresponding to $(i,j,k)$ is a hexagon.
\item The graph $\Gamma (U^{+ijk})$ is obtained from the 
graph $\Gamma (U)$ by superurban renewal at this hexagon. 
\item The variables associated with each face of $\Gamma (U)$
transform under superurban renewal according to the hexahedron 
recurrence~\eqref{hh1}--\eqref{hh4}, provided we interpret 
\begin{eqnarray*}
h(i,j,k) & = & A(i,j,k) \\
h^{(x)} (i,j,k) & = & A(i, j+1/2, k+1/2) \\
h^{(y)} (i,j,k) & = & A(i+1/2, j, k+1/2) \\ 
h^{(z)} (i,j,k) & = & A(i+1/2, j+1/2, k) \, .
\end{eqnarray*}
\end{enumerate}
$\noproof$
\end{proposition}

\subsection{Cubic corner graph and taut dimer configurations}

We now know that adding a cube to a downwardly closed stepped solid
corresponds to superurban renewal on the associated graph, which
corresponds to the use of the hexahedron recurrence to write the
top variable in terms of lower variables.  These representations
commute.  To state this more precisely, let $U_0$ be the stepped
solid coincident with the closed negative orthant.  The associated
graph $\Gamma (U_0)$ is called the \Em{cubic corner graph} and
is shown in Figure~\ref{cubiccorner}.
\begin{figure}[htbp]
\center{\includegraphics[scale=0.6]{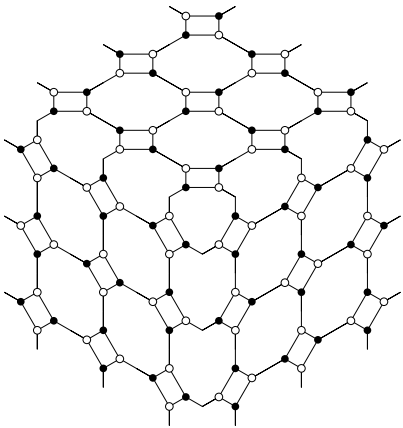}}
\caption{The cubic corner graph.}
\label{cubiccorner}
\end{figure}

Let $\lat$ be the lattice poset of all stepped solid subsets of 
$U_0$ containing all but finitely many cubes of $U_0$.  
For each $U \in \lat$, one may add a finite sequence of cubes
resulting in $U_0$.  Therefore, a finite sequence of superurban 
renewals represents $A(0,0,0)$ in terms of the variables labeling
faces and vertices of the stepped surface $\partial U$ that are
in the union of the removed lattice cubes.  Denote this set
of variables by $\init (U)$.

\begin{proposition}
$(i)$ The rational function $F$ representing $A(0,0,0)$ in terms of 
the variables in $\init$ is a Laurent polynomial.
$(ii)$ If $U'\subset U$ in $\lat$ and the representation
of each variable $w \in \init(U)$ in terms of variables in $\init(U')$
is substituted into $F$, the resulting Laurent polynomial is the 
representation of $A(0,0,0)$ in terms of variables in $\init(U')$.
\end{proposition}

\begin{proof}
By Proposition~\ref{pr:super}, the expression $F$ is obtained
by a sequence of superurban renewals.  By definition, each
of these is a sequence of six urban renewals, hence part~$(i)$
follows from Proposition~\ref{pr:urban}.  Part~$(ii)$ is a
consequence of the lack of relations among the variables
in any stepped surface.  
\end{proof}

Our combinatorial interpretations of these formulae 
take place on the associated graphs $\Gamma (U)$.
An example to keep in mind is $U_{-n}$, defined
to be those cubes of $U_0$ lying entirely within the
halfspace $\{ (x,y,z) \, : \, x + y + z \leq -n \}$.
This solid and its associated graph are illustrated 
for $n=-1$ (only the top cube removed)
in Figure~\ref{onecubegone}.  
\begin{figure}[htbp]
\centering
\includegraphics[scale=0.6]{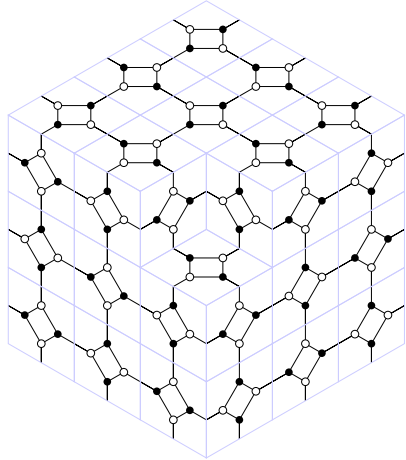}
\caption{After removing the topmost cube.}
\label{onecubegone}
\end{figure}
This solid is a subset
of $Z_{-n}$, the stepped solid defined by $x+y+z \leq -n$, whose associated graph is isomorphic to the 4-6-12 
graph of Figure~\ref{fig:4-6-12}.  The labels of $Z_{-n}$ are precisely
the points of $\Z^3$ at levels $-n-2, -n-1$ and $-n$ together
with the half integer points at level $-n-1$.  The hexahedron 
recurrence imposes no relations on this set of variables, 
hence from the point of view of determining $A(0,0,0)$ as a 
function of the variables in $\init (U_{-n})$, we might
equally well think of the initial variables as being all of 
those in $\Gamma (Z_{-n})$.

We define a double-dimer configuration $m_0$ on the cubic corner graph 
$\Gamma_0$ as in Figure \ref{initialconfig}.
\begin{figure}[htbp]
\center{\includegraphics[scale=0.6]{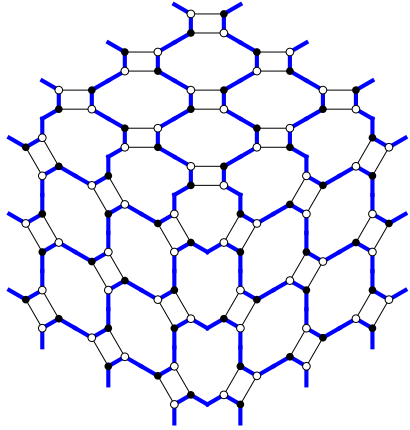}}
\caption{The initial double-dimer configuration $m_0$ on $\Gamma_0$.}
\label{initialconfig}
\end{figure}
This configuration $m_0$ plays the role of our initial configuration. 
This configuration has the following property.  If we erase a 
finite piece of $m_0$, there is a unique way to complete it to 
a double-dimer configuration which has the same boundary connections,
that is, connections between far-away points.  For $U \in \lat$,
we say that a double-dimer configuration on $\Gamma (U)$ is 
\Em{taut} if it has the same boundary connections as $m_0$, 
that is, it is identical to $m_0$ far from the origin and there 
is a bijection between its bi-infinite paths and those of $m_0$
which is the identity near $\infty$.  There are a finite number 
of taut configurations. See Figure \ref{Gamma1sample} for one such 
on $\Gamma (U_{-1})$.
\begin{figure}[htbp]
\center{\includegraphics[scale=0.8]{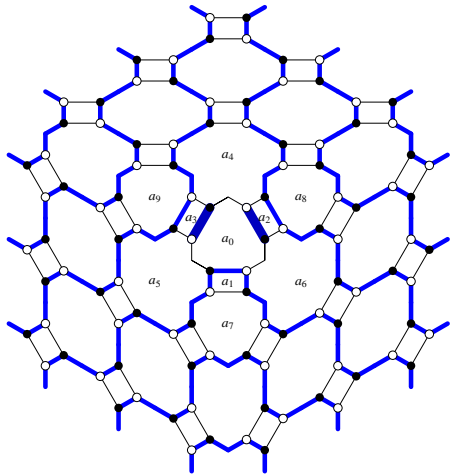}}
\caption{\label{Gamma1sample}A taut double-dimer configuration on 
$\Gamma (U_{-1})$. Doubled edges are thicker.}
\end{figure}

\section{Main formula} \label{sec:main}

Given a taut dimer configuration $m$, let $c(m)$ denote
the number of loops in $m$ and define $c(m;i,j,k) := L(i,j,k) 
- 2 - d(m;i,j,k)$ where $L(i,j,k)$ is the number of edges 
in the face $(i,j,k)$ and $d(m;i,j,k)$ is the number of dimers 
lying along the face $(i,j,k)$ in the matching $m$. 
Define the \emph{weight} of a taut configuration $m$ to 
be 
$$2^{c(m)} \prod_{(i,j,k) \in \init (U)} A (i,j,k)^{c(m; i,j,k)}$$
where $c(m)$ is the number of nontrivial loops in $m$.
In the
configuration $m_0$, all quadrilateral faces have 2 dimers
and all octagonal faces have 6 dimers, so the only face
$(i,j,k)$ with $c(m_0;i,j,k) \neq 0$ is the hexagonal face
which has 3 dimers and $c (m_0 ; 0,0,0) = 6 - 2 - 3 = 1$.
Any taut configuration differs from $m_0$ in finitely many
places, hence its weight has finitely many variables appearing in it.
For example, the configuration of Figure \ref{Gamma1sample} 
has weight $\displaystyle{\frac{a_4^2a_5a_6a_7}{a_0a_1a_2a_3}}$. 

\begin{theorem}\label{maincomb} 
Fix any $U \in \lat$ and let $\init (U)$ be the labels of 
$\Gamma (U)$.  Use the notation $m \preceq U$ to
signify that $m$ is a taut double-dimer configuration
on $\Gamma (U)$.  Then the representation of $A(0,0,0)$ 
as a Laurent polynomial in the variables in $\init (U)$ 
is given by
\begin{equation} \label{eq:main}
A(0,0,0) = \sum_{m \, \preceq \; U} 2^{c(m)} \prod_{(i,j,k) 
   \in \init (U)} A (i,j,k)^{c(m; i,j,k)} \, .
\end{equation}
Specializing to $U_{-n}$ and $A(i,j,k) = 1$ for
all $i,j,k$ with $-n-2 \leq i+j+k \leq -n$ gives the formula
$$A(0,0,0) = \sum_{m \, \preceq \; U_{-n}} 2^{c(m)} \, .$$
\end{theorem}

\begin{unremark}
These formulae are translation-invariant, once one accounts for the 
distinguished role played by $(0,0,0)$ in $\Gamma_0$.  
Let $(r,s,t)$ be any point of $\Z^3$ and $U$ be any lattice
solid obtained by removing finitely many cubes from the
top of the orthant $\{ x \leq r , y \leq s, z \leq t \}$.
Then $A(r,s,t)$ is a Laurent polynomial in the set $\init (U)$
of labels in $\Gamma (U)$ and~\eqref{eq:main} holds for $A(r,s,t)$
in place of $(0,0,0)$.
\end{unremark}

\begin{proof} We induct on $U$.  It is true for $U = U_0$:
there is one configuration, $m_0$, with $c(m_0 ; i,j,k) = 1$
if $i=j=k=0$ and zero otherwise.  The sum is therefore equal
to $A(0,0,0)$, yielding the identity $A(0,0,0) = A(0,0,0)$
which is the correct representation.  

For the induction to run, we need to see that the conclusion 
remains true if we remove a maximal cube.  This corresponds
to a superurban renewal, which is a composition of ordinary
urban renewals with additional vertex splittings and 
contractions, depending on which version of ordinary
urban renewal is used.  Under a vertex splitting, two faces 
increase in length by $2$ and get two additional dimers
on them.  Thus their contribution does not change.  A similar 
argument holds for vertex contraction.

Consider an urban renewal of type shown in Figure~\ref{Aurban}. 
There are several cases to consider, see Figure \ref{urbancheck},
depending on the various possible boundary connections. 
In the first case, the ratio of monomials on the right side 
and left side of the equation is
$$\frac{a_5^2a_1^{2-4}a_2^{2-4}a_3^{2-4}a_4^{2-4}}
   {2a_0^{-2}a_1^{-1}a_2^{-1}a_3^{-1}a_4^{-1}+
   a_0^{-2}a_1^{-2}a_2^{-2}a_3^0a_4^0+a_0^{-2}a_1^0a_2^0a_3^{-2}a_4^{-2}}
   = \frac{a_5^2a_0^2}{a_1a_2a_3a_4(2+\frac{a_3a_4}{a_1a_2} 
   + \frac{a_1a_2}{a_3a_4})} = 1 \, .$$
In the second case, the ratio is
$$\frac{a_5^1a_1^{2-4}a_2^{2-3}a_3^{2-3}a_4^{2-3}} 
   {a_0^{-1}(a_1^{-1}a_2^0a_3^{-1}a_4^{-1}+a_1^{-2}a_2^{-1}a_3^0a_4^0)}
   = \frac{a_0a_5}{a_1a_2+a_3a_4}=1 \, .$$
The remaining five cases are similar. These cover all possible cases 
(up to rotations).
\begin{figure}[htbp]
\center{\includegraphics[width=4in]{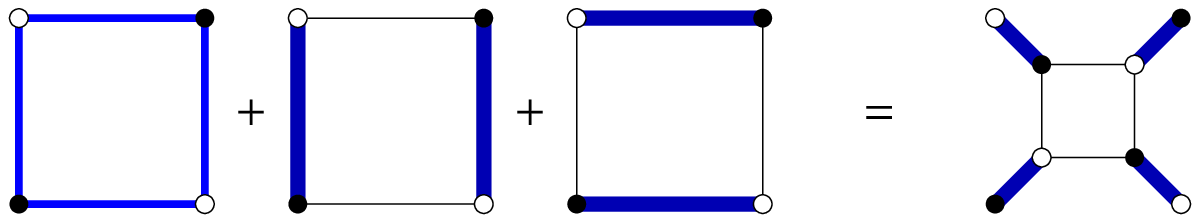}\\
\includegraphics[width=3in]{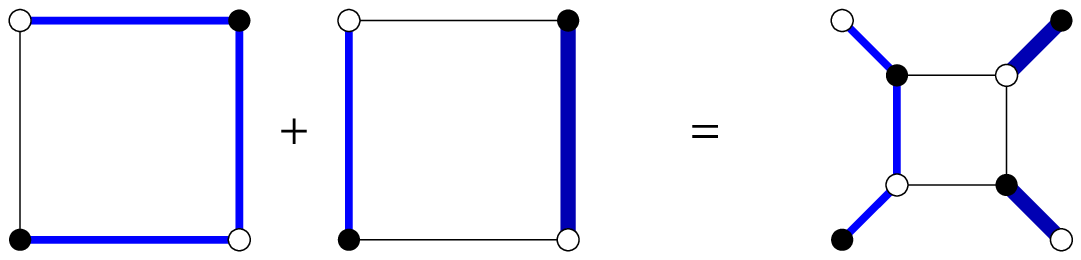}\\
\includegraphics[width=2in]{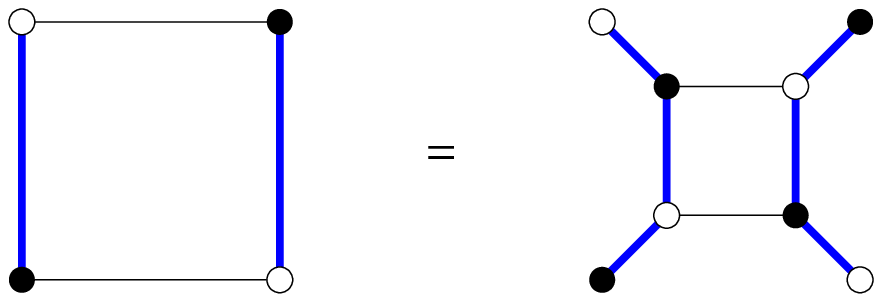}\\
\includegraphics[width=2in]{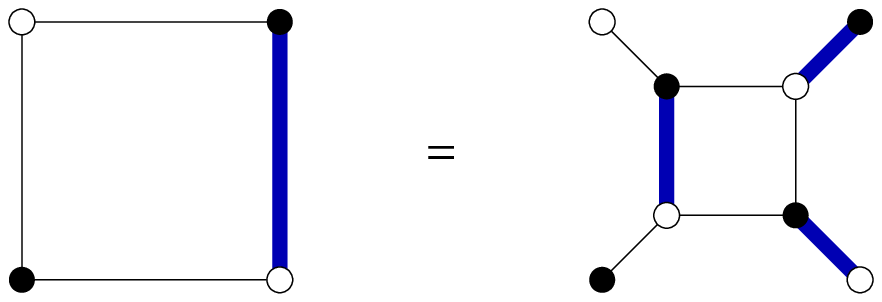}\\
\includegraphics[width=2in]{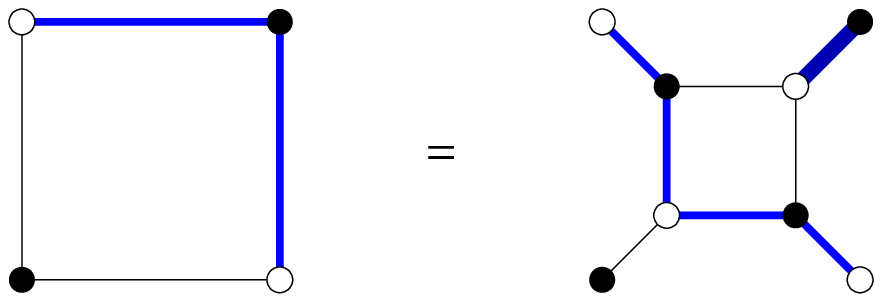}\\
\includegraphics[width=4in]{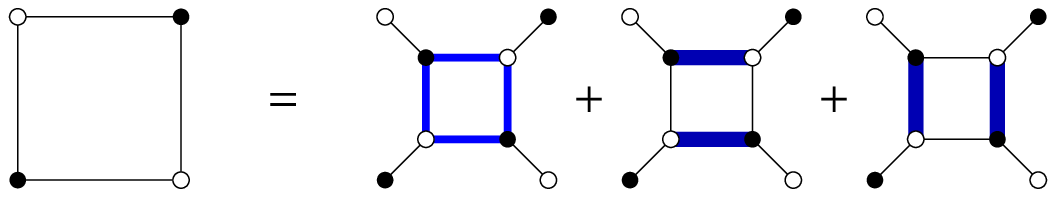}\\
\includegraphics[width=3in]{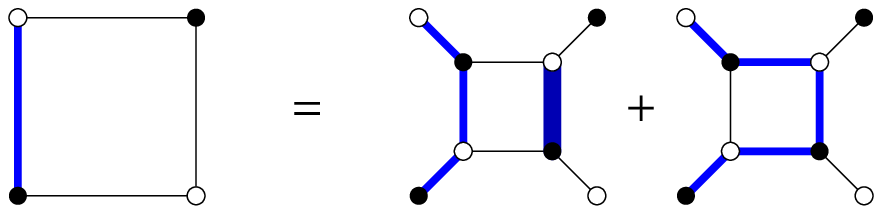}}
\caption{\label{urbancheck}Checking consistency of urban renewal formulas and $A$ monomials.}
\end{figure}
\end{proof}

\old{
The following definition gives some notation for interpreting
representations of the $A$ variables.  
\begin{definition}
For $(i,j,k) \in \Z^3$ with $i+j+k > 0$, let $\Gamma_{i,j,k}$
the graph associated with the stepped solid $U_{ijk}$ defined 
to consist of those lattice cubes all of whose points satisfy 
$$x \leq i , \;\; y \leq j ,  \;\; j \leq k , \;\;  x+y+z \leq 2 \, .$$ 
Let $\init_{ijk}$ denote the variables associated with 
faces of $\Gamma_{ijk}$ at levels~0, 1 and~2.  These
are precisely the variables occurring in the representation
of $A(i,j,k)$ in terms of the canonical initial variables.
\end{definition}
\clearpage
}

\section{Limit shapes}\label{ddlimitshape}

In this section we specialize values of the initial variables
in several natural ways and study the behavior of the resulting
ensembles.  Let $\ell : \Z^d \to \Z$ be a linear function which we 
regard as time.  A function $f : \Z^d \to \Z$ is said 
to be (spatially) isotropic if $f(v)$ depends only on $\ell (v)$.  
Solutions to recurrences that are isotropic and have the 
particularly simple form $f(v) = \gamma^{\ell (v)^k}$ lead to 
linear recurrences for the formal logarithmic derivatives 
$(\partial / \partial t) \log f(v,t)$; these have 
probabilistic interpretations and satisfy limit
shape theorems.  We begin by identifying the exponent $\delta$ 
for which such a solution exists.

\subsection{Isotropic solutions and homogeneity}
\label{ss:hom recs}

In order to discuss general $\Z^d$-invariant algebraic recurrences 
some notation is required.  Let $\A$ be a finite index set, let 
$\{ E_\alpha : \alpha \in \Z \}$ be a finite collection 
of multisets of elements of $\Z^d$, and let 
$\{ c_\alpha : \alpha \in \A \}$ be constants.  Denote 
the multiset union by $E := \bigcup_{\alpha \in \A} E_\alpha$.
For each $v \in \Z^d$, define a polynomial $P_{(v)}$ 
on indeterminates $\{ x_w : w \in \Z^d \}$ by the equation
\begin{equation} \label{eq:alg rec}
P_{(v)} = \sum_{\alpha \in \A} c_\alpha 
   \prod_{w \in E_\alpha} x_{v + w} \, .
\end{equation}
A lattice function $f : \Z^d \to \C$ is said to satisfy the 
algebraic recurrence  $P = \{ P_{(v)} : v \in \Z^d \}$ if at all $v\in\Z^d$, 
$P_{(v)}$ vanishes upon setting $x_w = f(w)$ for all 
$w \in \Z^d$.  Say that the recurrence is $k$-homogeneous 
with respect to the linear functional $\ell : \Z^d \to \Z$ 
if there are constants $\beta_0 , \ldots , \beta_k$ such 
that for each $\alpha \in \A$ and each $0\le j \leq k$, 
$$\sum_{w \in E_\alpha} \ell (w)^j = \beta_j \, .$$
For instance, the recurrence is 0-homogeneous if
$|E_\alpha| = \beta_0$, that is, if the polynomials
$P_{(v)}$ are homogeneous of degree $\beta_0$.  The
\Em{degree of homogeneity} of the recurrence $P$ is
the maximum $\delta$ such that $P$ is $\delta$-homogeneous.

\begin{example}
The octahedron recurrence~\eqref{eq:oct} is 1-homogeneous with 
respect to the height function $\ell (i,j,k) = j+k$.  To see this,
put the recurrence in the form of~\eqref{eq:alg rec}
by taking $E_1 = \{ (1,0,0) , (0,1,1) \}$, $E_2 = 
\{ (0,1,0) , (1,0,1) \}$ and $E_3 = \{ (0,0,1) , (1,1,0) \}$. 
The corresponding multisets of heights are $\{ 0 , 2 \}, 
\{ 1 , 1 \}$ and $\{ 1 , 1 \}$.  The sums of the
zero powers of the heights are the constant $\beta_0 = 2$.
The sums of the first powers of the heights are the constant $\beta_1 = 2$.
The sums of the squares of the heights are not constant
(4, 2 and~2) therefore the degree of homogeneity of the 
octahedron recurrence with respect to $\ell (i,j,k) = j+k$ is~1.
\end{example}

The degree of homogeneity of a recurrence tells us for
which power $\delta$ we can expect a solution to the 
recurrence of the form $f(v) = \gamma^{\ell (v)^\delta}$.
These solutions are the ones for which we can most
easily compute corresponding linear recurrences for the
the derivatives $(\partial / \partial t) f(v,t)$.
\begin{proposition} \label{pr:homog solutions}
Suppose the degree of homogeneity of the recurrence $P$ 
is $\delta-1$.  Suppose further that $\sum_{\alpha \in S} c_\alpha$
is nonvanishing for at least two equivalence classes $S = S_m
= \{ \alpha : \sum_{w \in E_\alpha} \ell (w)^\delta = m \}$.
Then there is a number 
$\gamma \neq 0$ such that $f(v) = \gamma^{\ell (v)^\delta}$ 
satisfies $P$ at every $v \in \Z^d$.  In fact 
$f(v) = \gamma^{q(\ell (v))}$ satisfies $P$ for
every monic polynomial $q$ of degree $\delta$.
\end{proposition}

\begin{proof}
Let $q$ be a monic polynomial of degree $\delta$ and define
$f(v) = f_q (v) = \gamma^{q(\ell (v))}$.  This satisfies
$P$ at $v$ if and only if
\begin{equation} \label{eq:q1}
\sum_{\alpha \in \A} c_\alpha \gamma^{\sum_{w \in E_\alpha} 
   q(\ell (v + w))} = 0 \, .
\end{equation}
For $j < \delta$, 
\begin{eqnarray*}
\sum_{w \in E_\alpha} \ell (v+w)^j 
& = & \sum_{w \in E_\alpha} (\ell (v) + \ell (w))^j \\
& = & \sum_{w \in E_\alpha} \sum_{i \leq j} \binom{j}{i} 
   \ell (v)^i \ell (w)^{j-i} \\
& = & \sum_{w \in E_\alpha} \sum_{i \leq j} \binom{j}{i} 
   \beta_{j-i} \ell (v)^i \\
& := & C(v)
\end{eqnarray*}
and factoring out $\gamma^{C(v)}$ from~\eqref{eq:q1} leaves
\begin{equation} \label{eq:q2}
\sum_{\alpha \in \A} c_\alpha \gamma^{\sum_{w \in E_\alpha} 
   (\ell (v + w))^\delta} = 0 \, .
\end{equation}
Again we may decompose the power $\ell (v+w)^\delta = 
(\ell(v) + \ell (w))^\delta$, this time arriving at
$$\sum_{w \in E_\alpha} \ell (v+w)^\delta =
   \sum_{w \in E_\alpha} \ell (w)^\delta + \sum_{1 \leq i \leq \delta} 
   \binom{\delta}{i} \beta_{\delta-i} \ell (v)^i := \ell (w)^\delta + D(v) \, .$$
Factoring out $\gamma^{D(v)}$ from~\eqref{eq:q2} leaves
\begin{equation} \label{eq:q3}
\sum_{\alpha \in \A} c_\alpha \gamma^{\sum_{w \in E_\alpha} 
   (\ell (w))^\delta} = 0 \, .
\end{equation}
This no longer depends on $v$, and is not a monomial
because we have assumed that $P$ is not $\delta$-homogeneous.
Therefore, there is at least one nonzero value of $\gamma$
for which~\eqref{eq:q3} holds, and this finishes the proof.
\end{proof}

\subsection{Recurrence for the derivative}

Suppose that each set $E_\alpha$ resides in the 
$\ell$-nonnegative halfspace $\Z^d_\ell := \{ v \in \Z^d : 
\ell (v) \geq 0 \}$, and suppose further that the values 
$\{ f(v) : v \in \Z^d_\ell \}$ all depend 
smoothly on some parameter $t$.  Consider the logarithmic
derivative $g(v) = (d/dt) \log f (v) = f'(v) / f(v)$,
which is defined whenever $f(v) \neq 0$.  We claim that
$g$ satisfies a linear recurrence with constant coefficients.

\begin{proposition} \label{pr:linear}
Suppose that $P$ is $(\delta-1)$-homogeneous but not $\delta$-homogeneous
and choose $\gamma$ according to Proposition~\ref{pr:homog solutions}
so that $f_0 (v) = \gamma^{\ell (v)^\delta}$ solves $P$.
Let $f(t,v)$ be smooth in $t$ with $f(0,v) = f_0 (v)$
and let $g(v) = (d/dt) \log f(t,v)|_{t=0}$.  Then
\begin{equation} \label{eq:lin rec}
\sum_{\alpha \in \A} c_\alpha' \sum_{w \in E_\alpha} \; g(v+w) = 0
\end{equation}
for all $v \in \Z^d_\ell$, where 
$$c_\alpha' = c_\alpha \gamma^{\ell (\alpha)}$$
and 
$$\ell (\alpha) := \sum_{w \in E_\alpha} \ell (w)^\delta \, .$$
\end{proposition}

\begin{proof}
Differentiating the recurrence at $v$ gives
\begin{eqnarray*}
0 & = & \left ( \frac{d}{dt} \right )_{t=0} 
   \sum_{\alpha \in \A} c_\alpha \prod_{w \in E_\alpha} x_{v+w} \\[1ex]
& = & \sum_{\alpha \in \A} c_\alpha 
   \left(\sum_{w \in E_\alpha} g(v+w)\right) \prod_{w \in E_\alpha} f(v+w) \\[1ex]
& = & \sum_{\alpha \in \A} c_\alpha \sum_{w \in E_\alpha} g(v+w) 
   \gamma^{\sum_{w \in E_\alpha} \ell(v+w)^\delta } \, .
\end{eqnarray*}
The sum $\sum_{w \in E_\alpha} \ell (v+w)^\delta$ is equal to 
$D(v) + \ell (\alpha)$, whence, factoring out $\gamma^{D(v)}$ 
from the last equation proves the proposition.
\end{proof}

The \Em{characteristic polynomial} of a recurrence
$\sum_{w \in E} b_w g(v+w) = 0$ is the Laurent polynomial
$\sum_{w \in E} b_w x^w$, where $x^w$ denotes the monomial
$x_1^{w_1} \cdots x_d^{w_d}$.  In particular, the 
characteristic polynomial of~\eqref{eq:lin rec} is 
$$H = H_P = \sum_{\alpha \in \A} 
   c_\alpha' \sum_{w \in E_\alpha} x^w \, .$$

\begin{example}
For the octahedron recurrence, let $E_1 := \{ (1,0,0) , (0,1,1) \}$,
$E_2 := \{ (0,1,0) , (1,0,1) \}$ and $E_3 := \{ (0,0,1) , (1,1,0) \}$
and denote these six vectors by $w_1 , \ldots , w_6$ respectively.
Then $\ell (1) = 4, \ell (2) = 2, \ell (3) = 2$.  The equation
for $\gamma$ is $\sum_\alpha c_\alpha' |E_\alpha| = 0$, which is
$2 \gamma^4 - 4 \gamma^2 = 0$, so $\gamma = \sqrt{2}$.
The linear recurrence on derivatives is then given by
$$2 [ g (v + w_1) + g(v + w_2)] - g(v + w_3) - g(v + w_4) - 
   g(v+w_5) - g(v + w_6) = 0 \, .$$
Dividing by~2, we see that the sum of the $x$ and $yz$ points
is equal to one half the sum of the other four points in any
elementary octahedron.  This recurrence has characteristic
polynomial $2(x+yz) - (y+z+xz+xy)$. 
\end{example}

\begin{example}[cube recurrence]
Let $w_1 , \ldots , w_8$ denote 
$$(0,0,0), (1,1,1) , (1,0,0), (0,1,1),
(0,1,0) , (1,0,1) , (0,0,1),(1,1,0)$$ respectively and
let $E_j = \{ w_{2j-1} , w_{2j} \}$ for $j=1,2,3,4$.  This
puts the cube recurrence~\eqref{eq:cube} in standard form with 
$c_1 = 1$ and $c_2 = c_3 = c_4 = -1$.  With $\ell (i,j,k) = i+j+k$, 
the cube recurrence has degree of homogeneity equal to $1$.  The
values of $\ell (\alpha) = \sum_{w \in E_\alpha} \ell (w)^2$ 
for $\alpha = 1,2,3,4$ are
$9,5,5,5$ and the resulting equation for $\gamma$ is
$\gamma^9 - 3 \gamma^5$ which has one positive solution 
$\gamma = 3^{1/4}$.  This leads to values for $c_\alpha'$
(we may divide everything by $3^{5/4}$) of $3, 1, 1$ and $1$
respectively.  Thus the recurrence for the derivatives is
given by
$$g(v) + g(v+w_2) = \frac{1}{3} \left ( \sum_{j=3}^8 g(v + w_j) \right )$$
and the characteristic polynomial for the
cube recurrence with respect to $\ell = i+j+k$ is
$$3(xyz + 1) - (x + y + z + xy + xz + yz) \, .$$
\end{example}

\subsection{Behavior of linear recurrences}

\subsubsection{Boundary conditions}
With the right boundary conditions, the linear 
recurrence~\eqref{eq:lin rec} will have tractable
asymptotics and a limiting shape.  This is assured
when the boundary conditions are such that~\eqref{eq:lin rec}
holds for all but finitely many $v \in \Z^d$.  The
most common way this arises is as follows.  Re-indexing 
if necessary, suppose that $0 \in E$ (recall $E = 
\bigcup_\alpha E_\alpha$), suppose that $\ell$ attains its
unique maximum there, and let $-m$ denote the minimum
value of $\ell$ on $E$.  The recurrence~\eqref{eq:lin rec}
determines $g(v)$ as a linear function of $\{ g(v + w) :
0 \neq w \in E \}$.  A canonical boundary condition is
to take $g(w) = 0$ when $\ell (w) < 0$, to take $g(0) = 1$
and to define $g$ everywhere else by the recurrence.  
In this case~\eqref{eq:lin rec} holds everywhere except
at the origin.  

Define a $d$-variable generating function 
$$F (x) := \sum_v g(v) x^v$$
where $x^v$ denotes the monomial $x_1^{v_1} \cdots x_d^{v_d}$.
Let $H$ denote the characteristic polynomial of the recurrence.
The fact that~\eqref{eq:lin rec} holds except at finitely
many points implies that the generating function $F$ satisfies 
$H F = G$ where $G$ is a Laurent polynomial.  

\begin{example}[octahedron recurrence, continued]
If we impose the octahedron recurrence everywhere except at 
the origin, setting the right-hand side of~\eqref{eq:lin rec} 
equal to~1 at the origin and setting $g(v) = 0$ when
$\ell (v) \leq 2$ except for $\ell (0,1,1) = 1$, then
$H F = x_2 x_3$, whence 
$$F = \frac{yz}{H} = \frac{yz}{2x + 2yz - y - z - xz - xy} \, .$$
The usual form of the octahedron recurrence differs from ours by
an affine transformation, whence one usually sees (e.g. for the 
Aztec Diamond creation rate generating function)
$$F(x,y,z) = \frac{z}{2 - (x + x^{-1} + y + y^{-1}) z + 2 z^2} \, ;$$
see~\cite{DGIPP} or~\cite[Section~4.1]{BaPe2011} for the generating 
function in this form.
\end{example}

\subsubsection{Coefficient asymptotics}
Laurent series for rational functions obey limit laws.  
A brief summary of the necessary background is as follows;
see, e.g.,~\cite[Chapter~7]{PW-book}.
Let $H$ be a Laurent polynomial vanishing at $(1 , \ldots , 1)$.
The amoeba of $H$ is the image in $\R^d$ of the complex zero set 
of $H$ under the log-modulus map $(z_1 , \ldots , z_d) \to 
(\log |z_1| , \ldots , \log |z_d| )$.  The components
of the complement of the amoeba are convex.  \emph{Assume that there is one
component, $B$, with the origin on the boundary}.  Any
rational function $G/H$ has a Laurent series $\sum_v a_v x^v$
convergent in the domain $\{ \exp (x + i y) : x \in B \}$.

The following further assumptions on $H$ put the
asymptotic behavior in the class of so-called cone points,
discussed in~\cite{BaPe2011} and~\cite[Chapter~11]{PW-book}.
Let $K$ be the convex cone of vectors for which all sufficiently
small positive multiples are in $B$.  The cone $K$ is a 
cone of hyperbolicity for $H$.  Let $K^*$ denote the dual cone 
$\{ w \in \R^d : \langle w , v \rangle \leq 0 ~\forall v \in K \}$.
We assume that $H$ is irreducible in the local ring at 
$(1 , \ldots , 1)$ and that $K^*$ has nonempty interior. 

\begin{theorem} \label{th:strictly minimal}
Let $F, G, H, B, K$ and $K^*$ be as above.  Assume that the zero 
set of $H$ touches the closed unit polydisk only at $(1,\dots,1)$.  Then
\begin{enumerate}
\item The coefficients $a_v$ decay exponentially in $|v|$ 
when $v \notin K^*$, the exponential rate being uniform
if $v / |v|$ is contained in a compact set disjoint from $K^*$.
\item If $G(1, \ldots , 1) =1$ then on $K^*$, 
$a_v \sim \Psi (v)$ (the ratio tends to $1$) where $\Psi$ is the inverse Fourier
transform of $1/\overline{H}$, the leading homogeneous part
of $H$ at $(1 , \ldots , 1)$.
\item The inverse Fourier transform $\Psi$ is homogeneous of degree 
$\deg (\overline{H}) - d$ .
\item If $G$ vanishes to degree $\delta$ at $(1 , \ldots , 1)$
then on $K^*$, $a_v \sim \Psi_G$, where $\Psi_G$ is a linear combination of 
partial derivatives of order $\delta$ of $\Psi$.
\end{enumerate}
If instead the zero set of $H$ intersects the unit torus in finitely 
many points then $K^*$ is the union of the dual cones at
each zero of $H$, the coefficients decay exponentially away
from $K^*$, and the leading asymptotic on $K^*$ is obtained
by summing $\Psi$ over the finitely many contact points of
the zero set with the unit torus.
\end{theorem}

\begin{proof}[Sketch of proof:]
The first part is essentially the Paley-Wiener Theorem.  It was
proved in the special case of cube groves (the cube recurrence)
in~\cite{PS}.  The general proof may be found in~\cite[Chapter~8]{PW-book}.
The second part is proved as~\cite[Lemma~6.3]{BaPe2011}.  There,
hypotheses are assumed to restrict the vanishing degree of
$H$ at $(1 , \ldots , 1)$ to~2, but in fact the proof is valid
for any degree.  The third part follows from the general theory
of inverse Fourier transforms, and the 
fourth part is~\cite[Proposition~6.2]{BaPe2011}, again removing the
restriction on the degree of vanishing of $H$ at $(1 , \ldots , 1)$.
\end{proof}

It is not possible to conclude that the support of $\Psi$ is all 
of the cone $K^*$.  More detailed information is given by computing
the critical sets.  The definition of the dual cones and critical sets 
are somewhat complicated, relying on hyperbolicity theory.  
We will not duplicate them here but will quote the result.
Again, the proof may be taken from~\cite{BaPe2011},
noting that the unit torus is a minimal torus, and that the finiteness
assumption on $W(v)$ satisfies the hypotheses of Theorem~5.8 there.
\begin{definition}[critical set]
For each $v$ in the interior of $K^*$, let $W(v)$ denote the set of
$z$ in the unit torus such that $v$ is in the 
dual cone to $H$ at $z$ as defined in~\cite[Definition~2.21]{BaPe2011}.
In particular when $\nabla H (z) \neq 0$, the point $z$ is in $W(v)$ exactly when 
$v$ is a scalar multiple of the logarithmic gradient $(z_1 \partial H / 
\partial z_1 , \ldots , z_d \partial H / \partial z_d) (z)$.
Let $K^\dagger \subseteq K^*$ denote the subset of $v$ such that
$W(v)$ is nonempty.  
\end{definition}

\begin{corollary} \label{cor:finite W}
Let $F, G, H, B, K, K^*$ and $K^\dagger$ be as discussed above.
Replace the hypothesis of Theorem~\ref{th:strictly minimal} 
that $H$ vanishes finitely often on the closed polydisk by 
the following hypotheses:
\begin{enumerate}[(i)]
\item The zero set of $H$ is disjoint from the open unit polydisk;
\item For each $v \in K^*$, the set $W(v)$ is finite.
\end{enumerate}
Whenever $H(z) = 0$, let $\overline{H}_z$ denote the homogeneous 
part of $H$ at $z$ and let $\Psi_z$ denote the inverse Fourier transform 
of $\overline{H}_z$.  Then for $v \in K^\dagger$, $a_v$ is
given asymptotically by the $\sum_{z \in W(v)} \Psi_z$.
The coefficients $a_v$ tend to zero exponentially rapidly when
$v$ goes to infinity and remains in any closed subcone disjoint
from the closure of $K^\dagger$.
$\noproof$
\end{corollary}

The phase boundary for a statistical mechanical system with generating
function $F$ is the boundary between exponential and non-exponential
decay of probabilities, which in most cases is $\partial K^\dagger$.
When $K^\dagger = K^*$, the phase boundary is equal to $\partial K^*$
and is not hard to compute: it is a component of the real algebraic 
hypersurface dual to the zero set of $\overline{H}$ in projective 
$(d-1)$-space.  When $d=3$, the boundary of
$K^*$ is an algebraic curve in $\RP^2$.  This curve may be
computed from $H (x,y,z)$ be setting $z = -ax - by$ and eliminating
$x$ and $y$ from the simultaneous polynomial equations
\begin{equation} \label{eq:dual}
H = 0  \; \; ; \; \; 
\frac{\partial H}{\partial x} = 0  \; \; ; \; \;
\frac{\partial H}{\partial y} = 0 \; .
\end{equation}

\subsection{Application to the hexahedron recurrence} \label{ss:hex}

We now apply the results of the last three subsections to
the hexahedron recurrence.  We have developed these results
independently of the specific recurrence for several reasons.
First, the methods are more general and it is good to see
them in the abstract.  Secondly, the hexahedron recurrence
is not $\Z^d$-invariant but invariant under a sublattice of 
finite index.  Because of this, it is not easy to see what
is going on if one begins with the hexahedron recurrence.
The exposition is clearest for general $\Z^d$-invariant 
recurrences, after which we may work the hexahedron recurrence
along similar lines.  We do not develop a general theory of
periodic recurrences for finite index sublattices because
the notation is even messier.  We first compute the isotropic
solutions, then compute the linear recurrences for the derivative, 
then prove limit shape results. 

\subsubsection{Isotropic solution}

Let $\ell (i,j,k) = i+j+k$.  An isotropic set of initial varibles 
$\init$ is given by the 4-6-12 graph: those $(i,j,k)$ with 
$0 \leq i+j+k \leq 2$.  The hexahedron recurrence, beginning 
with these variables, preserves the isotropy.  Therefore,
the solution will be described by constants $A_n , B_n : n \geq 0$
such that $f(i,j,k) = A_n$ for integer points with $i+j+k = n$
and $f(i,j,k) = B_n$ for half-integer points with $i+j+k = n+1$.
We will solve this general recursion, then specialize to solutions
of the form $A_n = \gamma^n, B_n = \kappa \gamma^n$.  
The initial conditions $A_0, A_1, A_2$ and $B_0$ and the hexahedron 
recurrence determine $A_n$ and $B_n$ for all positive $n$.
The recurrence becomes
{\scriptsize
\begin{eqnarray*} A_n & = & \frac{2 A_{n-2}^3 B_{n-3}^3+3 A_{n-3} A_{n-1} A_{n-2} B_{n-3}^3+A_{n-2}^6+3 A_{n-3}
   A_{n-1} A_{n-2}^4+3 A_{n-3}^2 A_{n-1}^2 A_{n-2}^2+A_{n-3}^3
   A_{n-1}^3+B_{n-3}^6}{A_{n-3}^2 B_{n-3}^3}\\
B_n &=& \frac{A_{n}^3+A_{n-1} A_{n} A_{n+1}+B_{n-1}^3}{A_{n-1} B_{n-1}}.
\end{eqnarray*}
}

The values of $A_3, B_1$ and $B_2$, are determined by the initial
conditions, but we may still use them in formulae for the 
remaining $A$ and $B$ values, leading to a mysteriously 
simple solution to the recurrence. 
$$A_4 = \frac{A_0A_3^2}{A_1^2},A_5=\frac{A_0^2A_3^4}{A_1^3A_2^2},
  A_6 = \frac{A_0^4A_3^6}{A_1^6A_2^3},\dots$$
$$B_3 = \frac{A_0B_1B_2^2}{A_2B_0^2} , 
  B_4 = \frac{A_0^2B_2^4}{A_2^2B_0^3},
  B_6 = \frac{A_0^4B_1B_2^6}{A_2^4B_0^6},\dots$$
and generally one can verify by induction that
$$A_n = \frac{A_0^{\lfloor (n-2)^2/4\rfloor}A_3^{\lfloor (n-1)^2/4\rfloor}}{
A_1^{\lfloor ((n-1)^2-1)/4\rfloor}A_2^{\lfloor ((n-2)^2-1)/4\rfloor}}.$$
$$B_{n} = \frac{A_0^{\lfloor (n-1)^2/4\rfloor}B_1^{\frac12(1-(-1)^n)}
   B_2^{\lfloor n^2/4\rfloor}}{A_2^{\lfloor (n-1)^2/4\rfloor}
   B_0^{\lfloor (n^2-1)/4\rfloor}} \, .$$
This can be written as
\begin{eqnarray}\label{gencaseAB1}A_{2n}
   & = &\frac{A_0^{(n-1)^2}A_3^{n^2-n}}{A_1^{n^2-n}A_2^{n^2-2n}} \\
A_{2n+1}&=&\frac{A_0^{n^2-n}A_3^{n^2}}{
A_1^{n^2-1}A_2^{n^2-n}} \\
B_{2n} & = & \frac{A_0^{n^2-n}B_2^{n^2}}{A_2^{n^2-n}B_0^{n^2-1}} \\
B_{2n+1} & = & \frac{A_0^{n^2}B_1B_2^{n^2+n}}{A_2^{n^2}B_0^{n^2+n}}
   \, . \label{gencaseAB4}
\end{eqnarray}

The simplest nontrivial solution to this recursion and the one in
the form of which we spoke earlier is
\begin{equation} \label{powerof3}
A_n=3^{n^2/2}, B_n = 2\cdot 3^{(n+1)^2/2} \, .
\end{equation}
There is another reasonably simple solution with a more direct
combinatorial meaning.  This is obtained by setting the initial 
variables $A_0, A_1, A_2, B_0$ all equal to~1.  This implies 
$A_3 = 14, B_1 = 3$ and $B_2 = 14$ and produces the result
\begin{eqnarray} 
A_{2n} & = & 14^{n(n-1)} \label{eq:14}\\
A_{2n+1} & = & 14^{n^2} \nonumber \\
B_{2n} & = & 14^{n^2} \nonumber \\
B_{2n+1} & = & 3 \times 14^{n(n+1)} \nonumber \, .
\end{eqnarray} 
Setting the initial $A$ variables equal to~1 amounts to setting
all the edge weights $\nu (e)$ equal to~1 in $\Gamma (U)$, as
dictated by change of variables~\eqref{eq:A-e}.
By Theorem~\ref{maincomb}, $A(0,0,0)$ counts taut 
double-dimer configurations of $\Gamma (U_{-n})$, with the
weight of $m$ counted as $2^{c(m)}$ when the canonical
initial variables are set
to~1.  By translation invariance, if $i+j+k = n+2$ and variables
at levels~0, 1 and~2 are set to~1, then $A(i,j,k)$ counts
taut double-dimer configurations in a graph isomorphic
to $\Gamma (U_{-n})$, again with weights $2^{c(m)}$.  
Evaluating $A(i,j,k) = A_{n+2} = 14^{(n/2)(n/2+1)}$
if $n$ is even and $14^{(n/2)(n/2+1) + 1/4}$ if $n$ is odd.
Thus we have proved:
\begin{corollary}
The number of taut double-dimer configurations of $\Gamma (U_{-n})$,
weighted by $2^{c(m)}$, is equal to 
$$14^{\frac{n}{2} \left ( \frac{n}{2} + 1 \right ) + \frac{1}{4}
   \delta_{{\rm odd}} (n)} \, .$$
\end{corollary}

\subsubsection{Recurrence for the derivative}

Let us interpret Proposition~\ref{pr:linear} in the 
context of statistical mechanical ensembles.  Suppose
that over the set $E = \bigcup_\alpha E_\alpha$, the
function $\ell$ has a minimum value of~0 and a 
maximum value of $J$.  Suppose further that $P$ is
an algebraic recurrence that determines the values
of $\{ f(v) : \ell (v) \geq J \}$ in terms of
initial values $\{ f(v) : 0 \leq \ell (v) < J \}$.
Consider $t = f(0 , \ldots , 0)$ to be variable
while all other initial conditions remain fixed
at $f(v) = \gamma^{\ell (v)^k}$.  Applying 
Proposition~\ref{pr:linear} gives the constant coefficient 
linear recurrence~\eqref{eq:lin rec} for the logarithmic 
derivatives of $f(v)$.  
Specializing further to the case $f(v) = A(v)$ for one of the 
Laurent recurrences we have studied, the monomials in the 
expression of $A(v)$ in terms of initial variables correspond 
to configurations, the value $A(v)$ is the partition function 
for all configurations.  The logarithmic derivative 
$\displaystyle{\frac{1}{A(v)} \; \frac{\partial A(v)}{\partial A(0,0,0)}}$ 
at the initial conditions $f(v) = \gamma^{\ell(v)^k}$ 
may be interpreted as the expected value of the
exponent on the term $A(0,0,0)$ in the statistical
mechanical ensemble in which the probability of 
the configuration $\xi$ is $M_\xi (v) / A(v)$ where
$M_\xi$ is the monomial corresponding to $\xi$.

We now apply this to the hexahedron recurrence and the
double-dimer ensemble.  We choose initial conditions~\eqref{powerof3}
rather than~\eqref{eq:14} because these correspond to
the solution $f(v) = \gamma^{\ell (v)}$ of 
Proposition~\ref{pr:homog solutions} with $\gamma(i,j,k) = i+j+k$.
The logarithmic derivative $g(i,j,k) = A(i,j,k)^{-1} 
\partial A(i,j,k) / \partial A(0,0,0)$ is the expected
number of dimers lying along the face at the origin in
a double-dimer configuration picked from all taut configurations
on $\Gamma_{ijk}$ according to the double-dimer measure $\mu_{dd}$
corresponding to the initial conditions~\eqref{powerof3}.
By translation invariance, this is the same as the
expected number of dimers lying along the face centered
at $(-i, -j, -k)$ on the graph $\Gamma (U_{-i-j-k})$.  We
may then ask about the limiting shape function, that is,
about the values of $g(i,j,k)$ as $n = i+j+k \to \infty$
with $(i/n, j/n, k/n) \to (\alpha_1 , \alpha_2, \alpha_3)$
in the 2-simplex.  

Taking the logarithmic derivative of the four recurrence 
relations and plugging in the initial conditions~\eqref{powerof3}
gives the linear system
\begin{eqnarray*}
g_{(123)} & = & -g + \frac13(g_{(1)} + g_{(2)} + g_{(3)}
   + g_{(23)} + g_{(13)} + g_{(12)})\\
g^{(x)}_{(1)} & = & \frac{1}{12} \left(-9 g + 4 g_{(1)} + g_{(2)} + g_{(3)} 
   + 3 g_{(23)}-4 g^{(x)} + 8 g^{(y)} + 8 g^{(z)}\right)\\
g^{(y)}_{(2)} & = & \frac{1}{12} \left(-9 g + g_{(1)} + 4 g_{(2)} + g_{(3)} 
   + 3 g_{(23)} + 8 g^{(x)}-4 g^{(y)} + 8 g^{(z)}\right)\\
g^{(z)}_{(3)} & = & \frac{1}{12} \left(-9 g + g_{(1)} + g_{(2)} + 4g_{(3)} 
   + 3 g_{(12)} + 8 g^{(x)} + 8 g^{(y)} - 4 g^{(z)}\right).
\end{eqnarray*}

As it happens, the first equation gives a self-contained recurrence
for the logarithmic derivatives at the integer points.  Not only
that, but the recurrence is recognizable as that arising in the 
cube recurrence~\eqref{eq:cube}.  In other words, letting
$F(x,y,z)=\sum g_{i,j,k}x^iy^jz^k$, we see that the solution
to the first recurrence above with boundary conditions 
$g(0,0,0) = 1$, $g(i,j,k) = 0$ for other points $(i,j,k)$
with $i+j+k \leq 0$ and satisfying the recurrence everywhere
except at $(-1-1-1)$, is 
$$F(x,y,z) = \frac{G(x,y,z)}{H(x,y,z)} = \frac{1}{1 + xyz - 
   \frac{1}{3} (x+y+z+xy+xz+yz)} \, .$$

\subsubsection{Limit shape}

This is the same as that satisfied by the cube grove
placement probabilities~\cite{PS}.  The boundary of the dual cone
is known as the ``arctic circle'', which is the inscribed circle 
in the triangular region $\{x+y+z=n,x,y,z\ge 0\}$.
Outside of this, the placement probabilities decay 
exponentially while inside the arctic circle they do not.
Inside, the limit function is homogeneous of degree $-1$
and is asymptotically equal to the inverse of the distance
to the arctic circle in the plane normal to the $(1,1,1)$ 
direction~\cite{BaPe2011}.  We can conclude from this that 
with high probability, a random configuration from 
$\Gamma_n$ is equal to $m_0$ outside a neighborhood of size 
$o(n)$ of the arctic circle and that there is positive local
entropy everywhere inside the arctic circle.

\subsection{General double-dimer shape theorems} \label{ss:GF}

Different periodic initial conditions lead to different 
limiting shapes.  In the general case we differentiate 
(\ref{hh1})-(\ref{hh4}) and use (\ref{gencaseAB1})-(\ref{gencaseAB4}) 
to get 8 linear equations, four for $i+j+k$ odd and four 
for $i+j+k$ even.  When $i+j+k$ is odd let $g_{i,j,k} = 
e_{i,j,k}a_{i+j+k}^{-1}$ and when $i+j+k$ is even let 
$h_{i,j,k}=e_{i,j,k}a_{i+j+k}^{-1}$.  Similarly define 
$g^{(x)},h^{(x)}$, and so forth.  This allows us to compute
the general solution assuming isotropy but not the simple
form of two geometric sequences $A_n , B_n$ as before.

\subsubsection*{A generic example}

The sequence $\{ A_n , B_n \}$ is determined by $a_0 , a_1 , a_2$ 
and $b_0$.  Let us start with a specific example (for the general case see below). Let $a_0=1,b_0=1,a_1=2,a_2=3$. 
Then $b_1=15,b_2=189,a_3=378$.  The linear system is
{\small
\begin{eqnarray*}g_{(123)}&=&\frac1{105}\left(84h+4(g_{(1)}+g_{(2)}+g_{(3)})+98(h_{(12)}+h_{(13)}+h_{(23)})-95(h^{(x)}+
h^{(y)}+h^{(z)})\right)\\
h_{(123)}&=&\frac1{42}\left(-33g+46(h_{(1)}+h_{(2)}+h_{(3)})+17(g_{(12)}+g_{(13)}+g_{(23)})-38(g^{(x)}
+g^{(y)}+g^{(z)})\right)\\
g^{(x)}_{(1)}&=&\frac1{210}\left(-126h+85g_{(1)}+g_{(2)}+g_{(3)}+84h_{(23)}-85h^{(x)}
+125 h^{(y)}+125 h^{(z)}\right)\\
g^{(y)}_{(2)}&=&\frac1{210}\left(-126h+g_{(1)}+85g_{(2)}+g_{(3)}+84h_{(13)}+125h^{(x)}
-85h^{(y)}+125 h^{(z)}\right)\\
g^{(z)}_{(3)}&=&\frac1{210}\left(-126h+g_{(1)}+g_{(2)}+85g_{(3)}+84h_{(12)}+125h^{(x)}
+125 h^{(y)}-85 h^{(z)}\right)\\
h^{(x)}_{(1)}&=&\frac1{15}\left(-9g+6g_{(23)}+14h_{(1)}+8h_{(2)}+8h_{(3)}-14g^{(x)}
+g^{(y)}+g^{(z)}\right)\\
h^{(y)}_{(2)}&=&\frac1{15}\left(-9g+6g_{(13)}+8h_{(1)}+14h_{(2)}+8h_{(3)}+g^{(x)}
-14g^{(y)}+g^{(z)}\right)\\
h^{(z)}_{(3)}&=&\frac1{15}\left(-9g+6g_{(12)}+8h_{(1)}+8h_{(2)}+14h_{(3)}+g^{(x)}
+g^{(y)}-14g^{(z)}\right).
\end{eqnarray*}
}
In terms of the generating functions, this is
{\tiny
$$
\begin{pmatrix}G\\H\\G^{(x)}\\G^{(y)}\\G^{(z)}\\H^{(x)}\\H^{(y)}\\H^{(z)}\end{pmatrix}=
\left(
\begin{array}{cccccccc}
 \frac{4}{105} (x y+z y+x z) & \frac{4 x y z}{5}+\frac{14}{15} (x+y+z) 
   & 0 & 0 & 0 & -\frac{19}{21} x y z & -\frac{19}{21} x y z & 
   -\frac{19}{21} x y z \\
 \frac{17}{42} (x+y+z)-\frac{11 x y z}{14} & \frac{23}{21} (x y+z y+x z) 
   & -\frac{19}{21} x y z & -\frac{19}{21} x y z & -\frac{19}{21} x y z 
   & 0 & 0 & 0 \\
 \frac{1}{210} x \left(\frac{1}{y}+\frac{1}{z}+\frac{85}{x}\right) 
   & \frac{1}{210} x \left(\frac{84}{y z}-\frac{126}{x}\right) 
   & 0 & 0 & 0 & -\frac{17 x}{42} & \frac{25 x}{42} & \frac{25 x}{42} \\
 \frac{1}{210} y \left(\frac{85}{y}+\frac{1}{z}+\frac{1}{x}\right) 
   & \frac{1}{210} y \left(\frac{84}{x z}-\frac{126}{y}\right) 
   & 0 & 0 & 0 & \frac{25 y}{42} & -\frac{17 y}{42} & \frac{25 y}{42} \\
 \frac{1}{210} \left(\frac{1}{y}+\frac{85}{z}+\frac{1}{x}\right) z 
   & \frac{1}{210} \left(\frac{84}{x y}-\frac{126}{z}\right) z 
   & 0 & 0 & 0 & \frac{25 z}{42} & \frac{25 z}{42} & -\frac{17 z}{42} \\
 \frac{1}{15} x \left(\frac{6}{y z}-9\right) & \frac{1}{15} x 
   \left(\frac{8}{y}+\frac{8}{z}+\frac{14}{x}\right) & -\frac{14 x}{15} 
   & \frac{x}{15} & \frac{x}{15} & 0 & 0 & 0 \\
 \frac{1}{15} y \left(\frac{6}{x z}-9\right) & \frac{1}{15} y 
   \left(\frac{14}{y}+\frac{8}{z}+\frac{8}{x}\right) & \frac{y}{15} 
   & -\frac{14 y}{15} & \frac{y}{15} & 0 & 0 & 0 \\
 \frac{1}{15} \left(\frac{6}{x y}-9\right) z & \frac{1}{15} 
   \left(\frac{8}{y}+\frac{14}{z}+\frac{8}{x}\right) z & \frac{z}{15} 
   & \frac{z}{15} & -\frac{14 z}{15} & 0 & 0 & 0
\end{array}
\right)
\begin{pmatrix}G\\H\\G^{(x)}\\G^{(y)}\\G^{(z)} \\
   H^{(x)}\\H^{(y)}\\H^{(z)}\end{pmatrix}+I_0,
$$
}
where $I_0$ represents the initial conditions.

The denominator of the generating functions is given by the determinant 
of $I-M$ where $M$ is the matrix on the RHS above.  This polynomial 
factors and the zero set has two components with equations
$$P_1=63 x^2 y^2 z^2-62(x^2 y z+x y^2 z + x y z^2)
   - (x^2y^2+x^2z^2+y^2z^2)+62(x y+x z+y z)+(x^2+y^2+z^2)-63$$
and
\begin{multline*}
   P_2=198 x^2 y^2 z^2-171(x^2 y^2 z+x^2 y z^2+x y^2 z^2)
   + 5(x^2 y^2+ x^2 z^2+y^2 z^2)+ 481(x^2 y z +x y^2 z+x y z^2) -  \\
   - 513 x y z-310(x y+x z+y z)-5(x^2+y^2+z^2)+315.
\end{multline*}

Evidently this is not irreducible.  However, writing the generating
function as $N / (P_1 P_2)$, there is a soft argument that the polynomial
$N$ in the numerator contains a factor of $P_2$ and therefore that the
generating function takes the form $F = G / P_1$.  To see this, 
observe by direct computation that $P_2$ has nontrivial intersection with 
$(-1,1)^3$.  Suppose that $N$ does not vanish on the intersection
of the zero set of $P_2$ with the open unit polydisk in $\C^3$.
Then the Taylor series for $F$ fails to converge at some point
in the open unit polydisk which means that the limsup growth
of the coefficients is exponential.  The probabilistic interpretation
contradicts this.  We conclude that $N$ vanishes on the intersection
of $P_2$ with the open unit polydisk, which is a variety of
complex codimension~1.  By irreducibility of $P_2$, we see that
$N$ vanishes on the whole zero set of $P_2$.  The upshot of all
this is that we may write $F$ in reduced form as $G / P_1$.  

Before checking the hypotheses we compute the dual curve to get
a picture of what we expect to find.  Translating by taking  
$x=1+X,y=1+Y$ and $z=1+Z$ and then taking 
the leading homogeneous (cubic) part gives
$$\overline{H} = 62 (X^2 Y + X Y^2 + X^2 Z  + Y^2 Z +X Z^2 + 
   Y Z^2)+ 132 X Y Z \, .$$
The arctic boundary is the dual of this cubic curve.  Computing
it as in~\eqref{eq:dual}, we arrive at a polynomial $P^* (a,b)$
defining an algebraic curve in $\CP^2$:
{\scriptsize
\begin{eqnarray*} P^* (a,b)  & = &  923521+5125974\,ba-3044572\,a{b}^{2}-2085370\,a{b}^{5}-3044572\,{b}^{3}a-3044572\,{a}^{2}b+45167\,{a}^{2}{b}^{4} \\
&& + 5125974\,{b}^{4}a+6191514\,
{a}^{2}{b}^{2}+2233364\,{b}^{3}{a}^{3}+45167\,{a}^{4}{b}^{2}-3044572\,
{a}^{2}{b}^{3}-2085370\,{a}^{5}b \\
&& -3044572\,{a}^{3}b+5125974\,{a}^{4}b-
3044572\,{b}^{2}{a}^{3}-2085370\,a-2085370\,b+45167\,{a}^{2}+45167\,{b
}^{2} \\
&& +45167\,{b}^{4}+2233364\,{b}^{3}+2233364\,{a}^{3}-2085370\,{b}^{5
}+45167\,{a}^{4}-2085370\,{a}^{5}+923521\,{b}^{6}+923521\,{a}^{6} \, .
\end{eqnarray*}
}
The zero set of $P^*$ contains two components in $\RP^2$.  These
are shown in Figure~\ref{examplearctic}.  The parametrization of
the curve $P^*$ above is via the representation of points in 
$\RP^2$ as $(a:b:1)$.  The picture is more symmetric in 
barycentric coordinates $(\alpha, \beta , 1-\alpha - \beta)$ 
where $a = \alpha / (1 - \alpha - \beta)$ and $b = \beta / 
(1 - \alpha - \beta)$.  Referring to figure~\ref{examplearctic},
we call the region inside the inner curve the ``facet'' and
the region between the two curves the ``annular region''.
The set $K^*$ is the union of these two regions.

\begin{figure}[htbp]
\center{\includegraphics[width=3in]{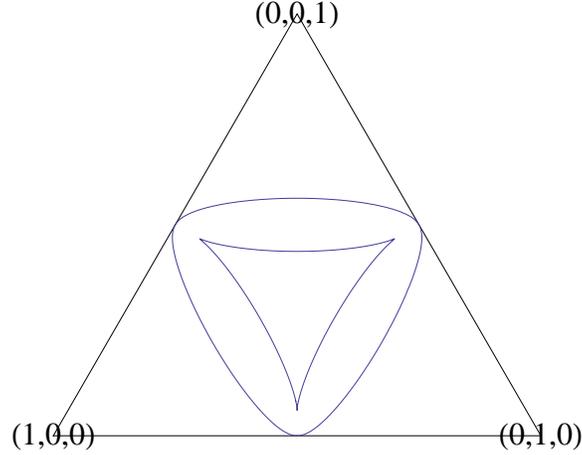}}
\caption{\label{examplearctic}The arctic boundary for the example. 
The curve is a homogeneous degree-$6$ curve.}
\label{fig:bdry}
\end{figure}

Next, we compute the inverse Fourier transforms at the points
of the intersection of the zero set of $H$ with the unit torus.
It is easily seen that these consist of $\pm (1,1,1)$ along with
a two-dimensional set of smooth points, where $\nabla H$ is nonvanishing;
we denote the smooth set by ${\cal V}$.  Because the degree of
$\overline{H}$ is~3 at $\pm (1,1,1)$ and~1 at any point of ${\cal V}$,
the inverse Fourier transform will have homogeneous degree~0
at $\pm (1,1,1)$ and $-2$ at smooth points.  

The IFT at $(1,1,1)$ is computed by an elliptic integral.
It is nonvanishing on the entire interior of $K^*$, 
varying over the annular region and remaining constant 
on the facet.  While everything else in this example
can easily be verified, this computation is not routine 
and will be detailed in forthcoming work~\cite{BaPe2013}.
The role of the contribution from $(-1, -1, -1)$ is to double 
the contribution to $a_v$ from $(1,1,1)$ when the parity of 
the integer vector $v$ is even and kill it when the parity is odd.  

Finally, to put this all together, we check the hypotheses of 
Corollary~\ref{cor:finite W}.  In fact, hypothesis~$(i)$ of
Corollary~\ref{cor:finite W} is guaranteed whenever the 
coefficients of $F$ are bounded.  To check hypothesis~$(ii)$,
we need only check that only finitely many points of ${\cal V}$
have a given logarithmic gradient.  This is true with the exception
of the projective directions $(1,1,0), (1,0,1)$ and $(0,1,1)$,
which are the midpoints of the sides of the triangle and are on
the outer boundary of the annular region.  Therefore, the hypotheses
are satisfied over the interior of $K^*$.  

Summarizing: for $v$ in the facet, $W(v) = \{ \pm (1,1,1) \}$,
while for $v$ in the annular region,
$W(v)$ is equal to the union of $\{ \pm (1,1,1) \}$ with a
finite set of points on the unit torus where $H$ vanishes
but its gradient does not.  This completes the verification 
of hypothesis~$(ii)$.  

We conclude that $K^\dagger = K^*$ is everything inside
the outer blue boundary in Figure~\ref{examplearctic}.
Outside $K^*$ there is exponential decay.  In the annular 
region, as $v \to \infty$ with $v / |v|$ tending to $\hat{v}$
and the parity of $v$ remaining even, the coefficient $a_v$ 
tends to a function $\Phi (\hat{v})$ given by an elliptic integral.
If $\hat{v}$ is in the facet, $a_v$ tends to a constant; due to
the three-fold symmetry, the constant must be $1/3$.  

\subsubsection*{In general}

In general there is a one-parameter family of curves, the
foregoing example being a generic instance.  The coefficients 
of the $8 \times 8$ array are rational functions of $a_0, b_0,
a_1$ and $a_2$; to get rid of subscripts, we denote these 
respectively by $a,b,c,d$.  Computing the characteristic
polynomial and factoring yields $P_1 P_2$ with 
\begin{eqnarray*}
P_1 & = & (C_1 + C_2) (x^2 y^2 z^2 - 1) 
   - C_1 (x^2 y^2 + x^2 z^2 + y^2 z^2 - x^2 - y^2 - z^2) \\
&&   \;\;\;\; - C_2 (x^2 y^2 + x^2 z^2 + x^2 y^2 + y^2 z^2 
   + x^2 z^2 + y^2 z^2 - xy-xz-yz) \\[1ex]
C_1 & := & a b^3 c d \\[1ex]
C_2 & := & b^6 + c^6 + 3 a c^4 d + 3 a^2 c^2 d^2 + 2 a b^3 c d
   + 2 b^3 c^3 + a^3 d^3 \, .
\end{eqnarray*}
In a neighborhood of the values $a=b=1, c=2, d=3$ from the 
worked example, the same argument shows there must be a factor
of $P_2$ in the numerator, so that the reduced generating function
is rational with denominator $P_1$; by analytic continuation, 
this is true for all parameter values.

Recentering at $(1,1,1)$ via the substitution $x = 1+X, y = 1+Y, z = 1+Z$
and taking the lowest degree homogeneous term at the origin yields
$$\overline{H} 
   = (1 - \theta) (X^2 Y + X^2 Z + Y^2 X + Y^2 Z + Z^2 X + Z^2 Y)
   + (2 + 6 \theta) X Y Z$$
where $\theta := C_1 / (C_1 + C_2)$.  We rewrite this as a constant times 
\begin{equation} \label{eq:lambda}
X^2 Y + X^2 Z + Y^2 X + Y^2 Z + Z^2 X + Z^2 Y + \lambda X Y Z
\end{equation}
where 
\begin{eqnarray*}
\lambda & = & \frac{2 + 6 \theta}{1 - \theta} \\[1ex]
& = & 
2\,{\frac {2\,{c}^{3}{b}^{3}+6\,acd{b}^{3}+{c}^{6}+3\,a{c}^{4}d+3\,{a}
^{2}{c}^{2}{d}^{2}+{a}^{3}{d}^{3}+{b}^{6}}{{a}^{3}{d}^{3}+{b}^{6}+3\,{
a}^{2}{c}^{2}{d}^{2}+3\,a{c}^{4}d+2\,acd{b}^{3}+{c}^{6}+2\,{c}^{3}{b}^
{3}}} \, .
\end{eqnarray*}
As $a,b,c,d$ vary over positive reals, the quantity $\lambda$ 
varies over the half-open interval $(2,3]$.  It reaches its 
maximum value when $a=1, b=2 \sqrt{3}, c = \sqrt{3}$ and $d=9$
(or at any scalar multiple of this 4-tuple of values) and
corresponds to the initial conditions~\eqref{powerof3}.

When $\lambda = 3$, the polynomial $\overline{H}$ factors as
$(X+Y+Z) (XY + XZ + YZ)$ and when $\lambda = 2$ it factors
as $(X+Y)(X+Z)(Y+Z)$.  However, for $2 < \lambda < 3$ this
polynomial is irreducible with the zero set looking 
like a cone together with a ruffled collar.  Figure~\ref{fig:collars}
shows two examples: on the left $\lambda = 5/2$ and on the right
$\lambda = 66/31$, a value much nearer to 2 which is the value
from the example with $a=b=1, c=2$ and $d=3$.
\begin{figure}[htbp]
\centering
\includegraphics[width=3in]{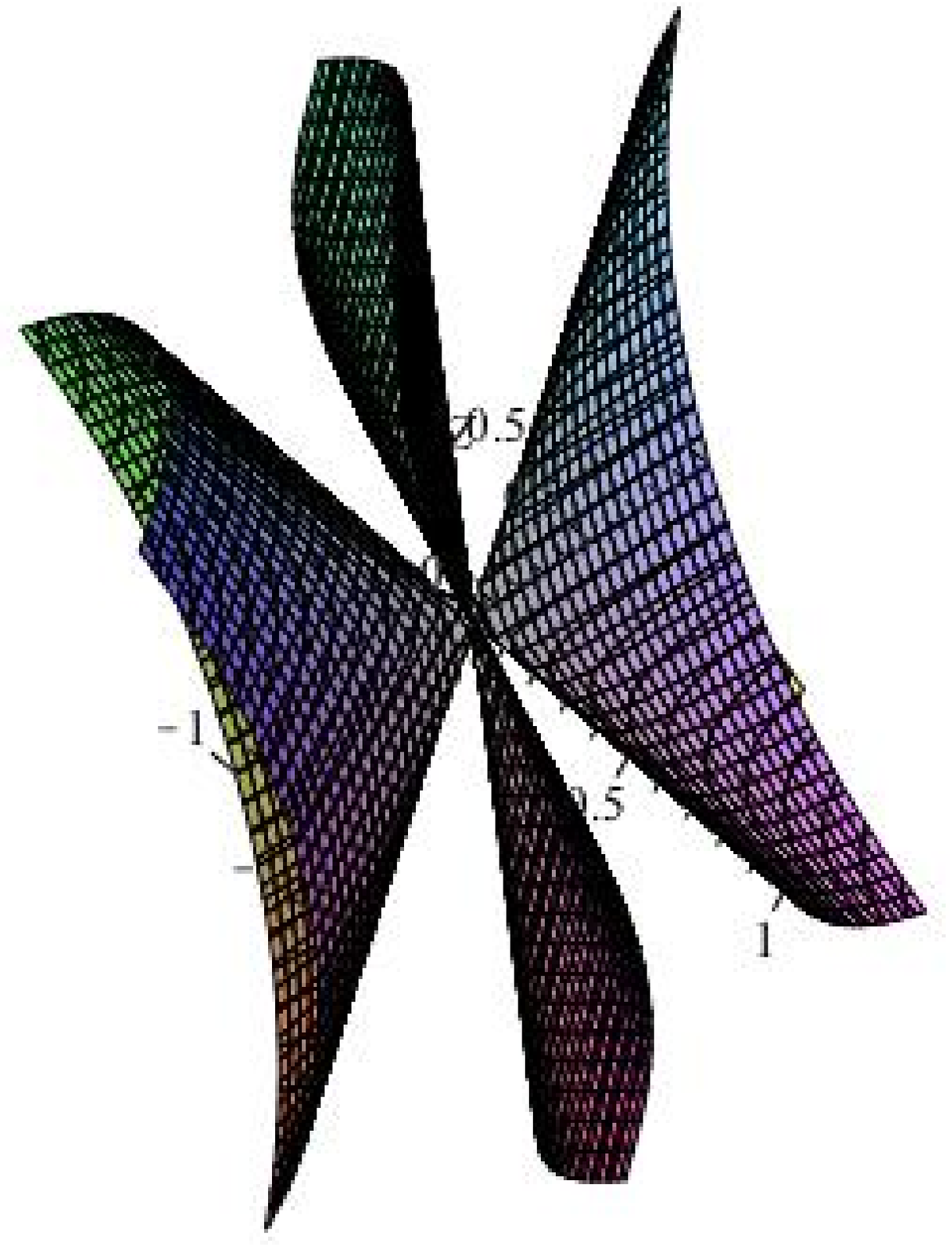}
\includegraphics[width=3in]{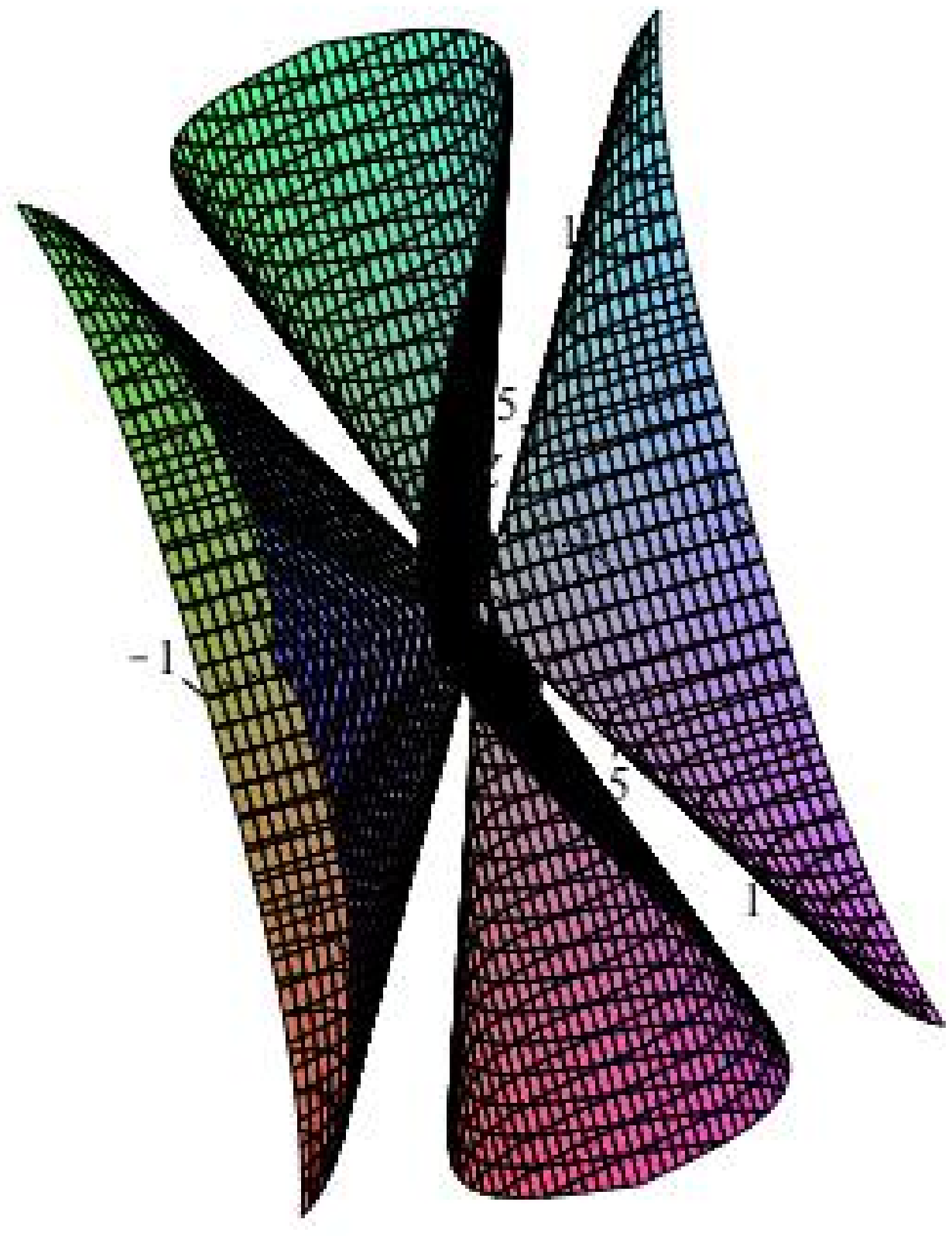}
\caption{\label{fig:collars} As $\lambda \to 2$ the collar becomes
more ruffled}
\end{figure}
At $\lambda = 3$ the dual shape is, as we have seen, the inscribed
circle; the facet region in this case is degenerate, having shrunk to
a point.  As $\lambda$ decreases from $3$ to $2$, the circle deforms 
to look more like an inscribed triangle and the facet grows, approaching
the outer curve.

\section{Ising model}\label{Ising}

In this section we will show how the Ising-Y-Delta move for the 
Ising model is a special case of the hexahedron recurrence.
We begin by recalling the definition of the Ising model.
Let $G = (V,E)$ be a finite graph with $c \, : \, E\to\R_+$ 
a positive weight function on edges.  The Ising model is a 
probability measure $\mu$ on the configuration space 
$\Omega = \{\pm1\}^{V}$.  A configuration of spins
$\sigma \in \Omega$ has probability
\begin{equation}
\label{isingdef}
\mu(\sigma) = \frac{1}{Z} \prod_{e = \{ v , v' \} \in E} 
   c(e)^{(1+\sigma(v)\sigma(v'))/2},
\end{equation}
where the product is over all edges in $E$ and the partition 
function $Z$ is the sum of the unweighted probabilities
$\prod c(e)^{(1 + \sigma (v) \sigma(v'))/2}$
over all configurations $\sigma$.  In other words, the probability 
of a configuration is proportional to the product of the
edge weights of those edges where the spins are equal.
The Ising model originated as a thermodynamical ensemble with 
energy function $H(\sigma) = -\sum_e \sigma (v) \sigma(v') J(e)$:
take $J(e) = (T/2) \log c(e)$ where $T$ is Boltzmann's
constant times the temperature.

\old{
If $G$ is a subgraph of a larger graph $\tilde G$ we often fix the 
spins outside and adjacent to $\G$; we then get a different measure
$\mu_B$ depending on these boundary spins $B$, defined by the same formula
(\ref{isingdef}) but taking the product over edges of $\G$ \emph{and} 
edges connecting vertices of $\G$ to vertices outside of $\G$. 
}

\subsection{Ising-Y-Delta move}

The Ising-Y-Delta move on the weighted graph $G = (V,E,c)$ 
transforms the graph the same way as does the Y-Delta move for
electrical networks but transforms the edge weights differently.
The transformation is depicted in Figure~\ref{YD}.  
\begin{figure}[htbp]
\center{\includegraphics[width=3in]{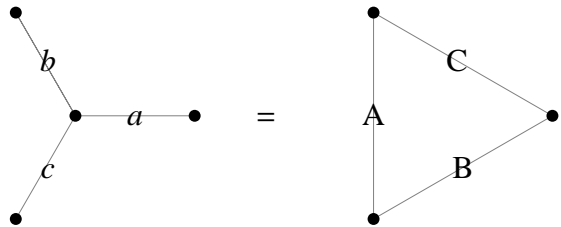}}
\caption{\label{YD}The Y-Delta move.}
\end{figure}
As is apparent, it converts a Y-shape to a triangle shape, 
or vice versa.  The old weights $(a,b,c)$ are replaced by 
weights $(A,B,C)$ defined by
\begin{eqnarray}
A&=&\sqrt{\frac{(abc+1)(a+bc)}{(b+ac)(c+ab)}}\label{Aa1}\\
B&=&\sqrt{\frac{(abc+1)(b+ac)}{(a+bc)(c+ab)}}\\
C&=&\sqrt{\frac{(abc+1)(c+ab)}{(a+bc)(b+ac)}}\label{Aa3}.
\end{eqnarray}

\begin{lemma}
Equations~\eqref{Aa1}--\eqref{Aa3} are the unique positive
solution to 
\begin{equation}
\label{projectiveeqn}[abc+1:a+bc:b+ac:c+ab] = [ABC:A:B:C]
\end{equation}
Referring to Figure \ref{YD}, if the Ising edge weights 
satisfy~\eqref{projectiveeqn} then the Ising-Y-Delta move 
preserves the measure $\mu$ in the sense that
there is a bijection on the space of configurations of the
graph $G$ and the graph $G'$ obtained from $G$ by applying
the Ising-Y-Delta transformation which maps the Ising
measure $\mu_G$ to the Ising measure $\mu_{G'}$.  The same
is true with boundary conditions.
\end{lemma}

\begin{proof} 
The first statement is a calculation. For the second statement, observe that up to a global sign 
there are four possible assignments of spins to the three 
vertices of the triangle (the three outer vertices of the Y). 
For example if the three spins are $+++$ then the edges of the 
triangle contribute weight $ABC$ to the weight of the configuration; 
in this case for the Y graph the central spin is either
$+$, in which case the contribution is abc, or $-$, in which case 
the contribution is $1$. Similarly the contributions for the 
$++-, +-+$, and $-++$ configuration are $C,B,A$ and $c+ab,b+ac,a+bc$, 
respectively.  As long as the quadruple of local contributions from 
the $Y$ is proportional to that of the $\Delta$ the measures
will be the same for the two graphs. 
\end{proof}

\subsection{Kashaev's relation}

Suppose that $\G = (V,E)$ is any planar graph with positive
weight function $c : E \to \R^+$.  Kashaev~\cite{Kash} showed 
how the space of edge weights for the Ising model on $\G$ can 
be parametrized differently using weights on the vertices and faces,
rather than the edges. This parametrization has the advantage that the
Y-Delta move has a simpler form in these new coordinates.  
Let $f$ be any positive function on vertices and faces of $\G$.
On each edge $e=vv'$ with adjacent faces $F,F'$,
let $b(e)$ be the ratio $b(e)=\frac{f(v)f(v')}{f(F)f(F')}.$
Kashaev associated the weight function $w$ with $f$ where
$w(e)$ is the positive solution of $(w - 1/w)^2 / 4 = b(e)$.

\begin{lemma}[\cite{Kash}] Let $f_0,\dots,f_7$ be the values 
at the faces and vertices involved in a Y-Delta transformation, 
as in Figure \ref{YDf}.  Then we have the identity 
\begin{multline}f_0^2f_7^2+f_1^2f_4^2+f_2^2f_5^2+f_3^2f_6^2
   - 2(f_1f_2f_4f_5+f_1f_4f_3f_6+f_2f_3f_5f_6)\\
   - 2f_0f_7(f_1f_4+f_2f_5+f_3f_6)-4(f_0f_4f_5f_6+f_7f_1f_2f_3)=0 \, .
   \label{Krec} \end{multline}
\end{lemma}

\begin{figure}[htbp]
\center{\includegraphics[width=3in]{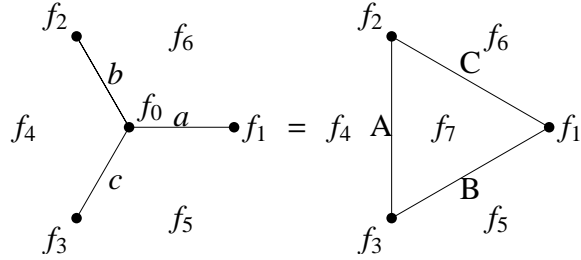}}
\caption{\label{YDf}The $f$ variables in the Y-Delta move.}
\end{figure}

\begin{proof} This is easy enough to check from (\ref{Aa1})-(\ref{Aa3}), 
with $\frac{(a-1/a)^2}{4}=\frac{f_0f_1}{f_5f_6}$, {\it etc}.
\end{proof}

We note that the remarkable formula (\ref{Krec}) has another origin: 
it is the algebraic identity relating the principal minors of a 
symmetric matrix, as follows.
\begin{lemma} Let $M$ be an $n\times n$ matrix and for $S\subset\{1,2,3\}$ 
let $M_S$ be the principal minor of $M$ which is the determinant of the 
matrix obtained from $M$ by removing rows and columns indexed by $S$. 
Then for the $8$ subsets of $S$ the identity (\ref{Krec}) holds with 
$$f_0 = M_\phi, f_1 = M_1, f_2 = M_2, f_3 = M_3, f_4 = -M_{23}, f_5 = -M_{13},
   f_6 = -M_{12}, f_7 = -M_{123} \, .$$
\end{lemma}

For an explanation of this fact, as well as the analogous facts for 
the Hirota and Miwa equations, see \cite{GK3}. 

\begin{proof}
One checks this easily for a $3\times 3$ matrix. For an 
$n\times n$ matrix $M$, recall that Jacobi's identity 
relates minors of $M$ with complementary minors of $M^{-1}$: 
$$\frac{M_S}{M_{\phi}} = (M^{-1})_{S^c}.$$
(In general there is a sign involved but for principal minors this sign is $+1$.)
The equation (\ref{Krec}) holds for the $3\times 3$ submatrix of $M^{-1}$ indexed by $S$; this implies that
it holds for $M$ for the complementary minors.
\end{proof}

When placed on a lattice, the relation (\ref{Krec}) has an interpretation
as a recurrence for  stepped surfaces.  Previously we associated 
a graph $\Gamma (U)$ with each stepped surface 
$\partial U$; now we associate another graph $\ising (U)$.  
The vertices of $\ising (U)$ are taken to be the even vertices 
of $\partial U$ and the edges of $\ising (U)$ are the diagonals 
of the faces of $\partial U$ whose endpoints are even.  
Because every face of $\partial U$ is a quadrilateral, the graph 
$\ising U$ is planar.  If $f : \Z^d \to \R^+$ is a positive function, 
define edge weight $w (e)$ on an edge $e$ of $\ising (U)$
to be the positive solution to $(w - 1/w)^2 / 4 = b$ where
$b = f(v) f(v') / (f(u) f(u'))$, where $e = \{ v , v' \}$ and
where $u$ and $u'$ are the other two vertices of the face of
$\partial U$ on which $e$ lies.  The previous lemma results
in the following lattice relation, known as {\em Kashaev's 
difference equation}.  

\begin{lemma}
Let $U \subseteq U'$ be stepped solids differing by a single cube.
\begin{enumerate}
\item The graph $\ising (U')$ differs from $\ising (U)$ by
a Y-Delta move: Y to Delta if the bottom vertex of the added 
cube was even and Delta to Y otherwise.
\item If $e$ is a weight function on the edges of $\ising (U)$,
extended by the Ising-Y-Delta relations to the edges of $\ising (U')$,
and if $f$ is a function on the vertices of $\Z^d$ inducing $e$
on the edges of $\ising (U)$ and $\ising (U')$ then at the 
eight vertices of the added cube, $f$ satisfies the relations
(\ref{KashaevEQ}).
\end{enumerate}
$\noproof$
\end{lemma}

Kashaev's relation is almost a recurrence: $f_{(123)}$ is determined
from the other seven values up to the choice of root of a 
quadratic equation.  It turns out there is a canonical choice.

\begin{proposition}
Let 
$$
X = \sqrt{ff_{(23)}+f_{(2)}f_{(3)}},\quad
Y=\sqrt{ff_{(13)}+f_{(1)}f_{(3)}},\quad
Z=\sqrt{ff_{(12)}+f_{(1)}f_{(2)}}.
$$
Then the recurrence~\eqref{KashaevEQ} 
can be written
\begin{eqnarray}\label{Krecreduced1}
X_{(1)}&=&\frac{f_{(1)}X+YZ}{f}\\Y_{(2)} & = &
   \frac{f_{(2)}Y+XZ}{f}\\Z_{(3)}&=&\frac{f_{(3)}Z+XY}{f}
   \label{Krecreduced3}\\
f_{(123)} & = & \frac{2f_{(1)}f_{(2)}
   f_{(3)}+ff_{(1)}f_{(23)}+ff_{(2)}f_{(13)}+ff_{(3)}
   f_{(12)}+2XYZ}{f^2} \, . \label{Krecreduced4}
\end{eqnarray}
$\noproof$
\end{proposition}
The proof is a simple verification.

\subsection{Embedding Kashaev's recurrence in the hexahedron recurrence}

The proof
of all results in this section are straightforward substitutions
and are omitted.

\begin{proposition} \label{Isingtodimer}
Suppose $f : \flabel \to \C$ satisfies the following relation
for integer $(i,j,k)$:
\begin{eqnarray*}
f(i+1/2,j+1/2,k)^2 & = & f(i,j,k) f(i+1,j+1,k) + f(i,j+1,k) f(i+1,j,k)
   \\
f(i+1/2,j,k+1/2)^2 & = & f(i,j,k) f(i+1,j,k+1) + f(i,j,k+1) f(i+1,j,k)
   \\
f(i,j+1/2,k+1/2)^2 & = & f(i,j,k) f(i,j+1,k+1) + f(i,j,k+1) f(i,j+1,k)
   \, .
\end{eqnarray*}
Then $f$ satisfies the Kashaev relation~\eqref{KashaevEQ} 
at integer points if $f$ satisfies the hexahedron 
relations~\eqref{hh1}--\eqref{hh4}, where as usual we interpret 
$h = f$, $h^{(x)} = f_{(0,1/2,1/2)}$, $h^{(y)} = f_{(1/2,0,1/2)}$ 
and $h^{(z)} = f_{(1/2,1/2,0)}$. 
$\noproof$
\end{proposition}

We obtain the Ising-Y-Delta recurrence as a corollary.  
\begin{corollary}
Suppose the initial conditions for $f$ at the vertices of 
a stepped surface $\partial U$ are real and positive.  
Define $f$ on the $z$-faces of the stepped surface by 
$$f(i+1/2,j+1/2,k) := \sqrt{f(i,j,k) f(i+1,j+1,k) + 
f(i,j+1,k) f(i+1,j,k)}$$ and similarly for the $x$- and 
$y$-faces, always taking the positive square root.  Then
the values produced by the hexahedron recurrence at all
points above the stepped surface, restricted to integer points,
yield the Ising-Y-Delta recurrence.
$\noproof$
\end{corollary}

\begin{unremark} 
Another way to say this is that the equation~\eqref{Krec} is a 
special case of~\eqref{superurbanrec}.
Setting $a_1=\sqrt{a_5a_6+a_0a_7}, a_2=\sqrt{a_4a_6+a_0a_8}$ and
$a_3=\sqrt{a_4a_5+a_0a_9}$, the recurrence~\eqref{superurbanrec}
becomes~\eqref{Krec} after relabelling variables to correspond to 
the same geometric positions: $f_0=a_0, f_1=a_4, f_4=a_7$ etc.
\end{unremark}

As in the dimer case let us consider initial conditions 
on the stepped surface defined by $0 \leq i+j+k \leq 2$ 
(take $U$ to be the lattices cubes lying entirely within 
$\{ (x,y,z) : x+y+z \leq 2 \}$).  Recall that $\Gamma (U)$
was the 4-6-12 graph.  It is easy to see that $\ising (U)$ 
is the regular triangulation.  Indeed, there is a vertex
of $\ising (U)$ at the center of each dodecagon of $\Gamma (U)$
and an edge of $\ising (U)$ connecting dodecagons that share
a quadrilateral neighbor.  The function $f$ now takes values 
on the vertices and faces of this triangulation and the 
$X,Y,Z$ values lie on the three directions of edges. 
Starting with initial data of the $f$ values on $\{0\le i+j+k\le 2\}$, 
one can determine $X_{i,j,k},Y_{i,j,k},Z_{i,j,k}$ on the set 
$i+j+k=0$. From~\eqref{Krecreduced1}--\eqref{Krecreduced3} 
one can then determine the values of $X,Y,Z$ on $i+j+k=1$; then 
from~\eqref{Krecreduced4} one can determine the values of $f$ on 
$i+j+k=3$, and so on; in this way one determines $X,Y,Z,f$ 
on all planes $i+j+k\ge 0$.
Theorem \ref{maincomb} has the following consequence for the $f_{i,j,k}$.

\begin{theorem} $f_{i,j,k}$ is a Laurent polynomial in the 
initial variables $\{f_{i,j,k}\}_{0\le i+j+k\le 2}$ and
$\{X_{i,j,k},Y_{i,j,k},Z_{i,j,k}\}_{i+j+k=0}$.  The $X,Y,Z$ 
variables only appear with power $1$. 
\end{theorem}

\begin{proof}
Let $\{ a_{ijk} \}$ denote the initial conditions for the hexahedron recurrence
(before specialization), thus 
$f_{ijk} = a_{ijk}$ for $0 \leq i+j+k \leq 2$.
Take a monomial $M$ in the Laurent expansion of $a_{i,j,k}$. Let
$a_0$ be a quadrilateral variable; it occurs in $M$ with exponent in $[-2,2].$

There are several cases to consider.
If $a_0$ occurs with power $2$, we can replace it with $a_1a_3+a_2a_4$
where $a_1,a_2,a_3,a_4$ are the 
four faces in cyclic order adjacent to $a_0$. Then the monomial $M$ becomes a sum of two monomials
not involving $a_0$. 

If $a_0$ appears with degree $-1$, there is another monomial 
$M'$ of $a_{i,j,k}$ which pairs with it, in the sense that 
$M/M'=\frac{a_1a_3}{a_2a_4}$ (these two monomials 
correspond to the two possible configurations of double-dimers 
which have three edges lying along the quad face at $a_0$, 
and use the same two edges joining the quad face to adjacent faces: as in line $2$ or Figure \ref{urbancheck}).  
The sum of these two monomials is, up to monomial factors $M^*$ 
not involving $a_0$, $M+M'=M^*(a_1a_3+a_2a_4)/a_0 = M^*a_0.$ 
Upon replacing $a_0$ by $\sqrt{a_1a_3+a_2a_4}$
this now becomes a monomial in which $a_0=\sqrt{a_1a_3+a_2a_4}$ occurs with power $+1$. 

If $a_0$ which appears in $M$ with degree $-2$, there are two other monomials 
$M',M''$, which in the appropriate order have the ratios 
$\displaystyle{[2:\frac{a_1a_3}{a_2a_4}:\frac{a_2a_4}{a_1a_3}]}$ 
(these correspond to the three possible configurations of double-dimers 
which have four edges lying along the quad face at $a_0$, as in the first line of  Figure \ref{urbancheck}).
The sum of these is 
$$M+M'+M''=M^*\frac{\displaystyle{2+\frac{a_1a_3}{a_2a_4}
   + \frac{a_2a_4}{a_1a_3}}}{a_0^2} 
   = M^*\frac{a_0^2}{a_1a_2a_3a_4}$$
and upon the substitution $a_0^2=a_1a_3+a_2a_4$ this is a sum of 
two monomials not involving $a_0$. 

Once all these substitutions (and groupings) are done, $a_0$ only appears in the numerator and has degree $1$ or $0$.
We can similarly regroup terms for all the other quad variables: since no two quad faces are adjacent the groupings
``commute" in the sense that they can be done in any order.

After regrouping all quad variables, we see that $f_{i,j,k}$ (the specialization of $a_{i,j,k}$) 
is a Laurent polynomial, with positive coefficients, in the initial
variables $f,X,Y,Z$, and with the $X,Y,Z$ (the quad variables) appearing in the 
numerator only and of degree $0$ or $1$. 
\end{proof}

\subsection{Open question}

What are the natural combinatorial structures counted by the monomials in $f_{i,j,k}$?
Using Proposition~\ref{Isingtodimer} it appears possible 
(although we have not succeeded) 
to get an interpretation of the monomials in the expansion of 
$f_{i,j,k}$ in terms of collections of double-dimer covers of 
$\Gamma_{i+j+k}$. 

Note however that when we apply the substitutions of the proof of 
that proposition, it is possible that the new monomials are not 
distinct, and combine to make monomials of $f_{i,j,k}$ in other ways. 
This is already true of $a_{1,1,1},$ in which $9$ monomials
collapse into a single monomial of $f_{1,1,1}$.

\subsection{Solving the recurrence for spatially isotropic initial conditions}

There is a three-parameter family of initial conditions
for which $f_{(i,j,k)} = f_{i+j+k}$ only depends on $i+j+k$: 
choose values $f_0 = a, f_1 = b$ and $f_2 = c$ arbitrarily
and set the initial $X$, $Y$ and $Z$ variables all equal to 
$\sqrt{ac + b^2}$.  The recurrence~\eqref{KashaevEQ} can be 
solved for $f_{n+3}$ as a function of $f_n,f_{n+1},f_{n+2}$,
giving
$$f_{n+3}=\frac{2 f_{n+1}^3+3 f_n f_{n+2} f_{n+1}+2 
   \sqrt{f_{n+1}^6+3 f_n f_{n+2} f_{n+1}^4+3 f_n^2 f_{n+2}^2 f_{n+1}^2
   + f_n^3 f_{n+2}^3}}{f_n^2} \, .$$
This can be solved explicitly for $f_n$ as a function of $a, b$ and $c$.
Letting $R := ac/b^2$ and
$$S=\frac{2 (R+1)^{3/2}+3 R+2}{R^2}$$
gives
\begin{equation}
\label{fn}
f_n = a^{1-n} b^n R^{\left\lfloor \frac{n^2}{4}\right\rfloor } 
   S^{\left\lfloor \frac{1}{4} (n-1)^2\right\rfloor } \, ,
\end{equation}
or in terms of alternating parity,
\begin{eqnarray}
f_{2n} & = & a^{1-2n}b^{2n}R^{n^2}S^{n^2-n}\label{f2n} \\[1ex]
   f_{2n+1} & = & a^{-2n}b^{2n+1}R^{n^2+n}S^{n^2} \, . \label{f2n+1}
\end{eqnarray}
For example when $a=1,b=\sqrt{3},c=9$ we have $R=S=3$ and  
\begin{equation}
\label{equil}
f_{i,j,k}=f_{i+j+k}=3^{(i+j+k)^2/2}.
\end{equation}
One can likewise compute limit shapes for $f_{i,j,k}$ as we did for 
$a_{i,j,k}$ although without a probabilistic interpretation their 
meaning is dubious.

\end{document}